\documentclass[runningheads]{llncs}
\usepackage[a4paper, total={6in, 8in}]{geometry}
\usepackage{cite}
\usepackage{mathtools}
\usepackage{amsmath,amssymb,amsfonts}
\usepackage{algorithmic}
\usepackage{graphicx}
\usepackage{textcomp}
\usepackage{xcolor}
\usepackage{etoolbox} 
\usepackage[vlined,boxed,commentsnumbered, ruled, linesnumbered]{algorithm2e}
\usepackage{commath}
\usepackage{url}
\usepackage{float}
\usepackage{subfig}
\usepackage{soul}
\usepackage[english]{babel}
\usepackage{amssymb}
\usepackage{esvect}
\usepackage{booktabs} 
\usepackage{gensymb}
\usepackage{tablefootnote}
\usepackage{booktabs} 
\usepackage{etoolbox} 
\usepackage[vlined,boxed,commentsnumbered, ruled, linesnumbered]{algorithm2e}
\newtheorem{mydef}{Definition}
\def\BibTeX{{\rm B\kern-.05em{\sc i\kern-.025em b}\kern-.08em
		T\kern-.1667em\lower.7ex\hbox{E}\kern-.125emX}}
\usepackage{setspace}
\usepackage{breakcites}
\usepackage{enumerate}
\usepackage{makecell}
\usepackage{float}
\usepackage{listings}
\usepackage{tikz}
\usetikzlibrary{trees}
\usepackage{paralist}
\usepackage{ragged2e}
\usepackage[title]{appendix}
\usepackage{vwcol}  
\usepackage{paracol}
\makeatletter
\newcommand{\removelatexerror}{\let\@latex@error\@gobble}
\makeatother
\begin{document}
	\title{Local Differential Privacy for Federated Learning}
	%
	%
	\author{M.A.P. Chamikara\inst{1,2}\orcidID{0000-0002-4286-3774} \and
		Dongxi Liu \inst{1} \and
		Seyit Camtepe\inst{1} \and
		Surya Nepal\inst{1} \and
		Marthie Grobler\inst{1} \and
		Peter Bertok\inst{3} \and
		Ibrahim Khalil\inst{3}
	}
	\authorrunning{M.A.P. Chamikara et al.}
	%
	\institute{CSIRO's Data61, Australia \and
		Cyber Security Cooperative Research Centre, Australia
		\and
		RMIT University, Australia
	}
	\maketitle              
	\begin{abstract} 
		Advanced adversarial attacks such as membership inference and model memorization can make federated learning (FL) vulnerable and potentially leak sensitive private data. Local differentially private (LDP) approaches are gaining more popularity due to stronger privacy notions and native support for data distribution compared to other differentially private (DP) solutions. However, DP approaches assume that the FL server (that aggregates the models) is honest (run the FL protocol honestly) or semi-honest (run the FL protocol honestly while also trying to learn as much information as possible). These assumptions make such approaches unrealistic and unreliable for real-world settings. Besides, in real-world industrial environments (e.g., healthcare), the distributed entities (e.g., hospitals) are already composed of locally running machine learning models (this setting is also referred to as the cross-silo setting). Existing approaches do not provide a scalable mechanism for privacy-preserving FL to be utilized under such settings, potentially with untrusted parties. This paper proposes a new local differentially private FL (named LDPFL) protocol for industrial settings. LDPFL can run in industrial settings with untrusted entities while enforcing stronger privacy guarantees than existing approaches. LDPFL shows high FL model performance (up to 98\%) under small privacy budgets (e.g., $\varepsilon = 0.5$) in comparison to existing methods. 
		
		
	\end{abstract}

	\section{Introduction}
	\label{introsection}
	Server-centric machine learning (ML) architectures cannot address the massive data distribution in the latest technologies utilized by many industries (cross-silo setting), including healthcare and smart agriculture. Besides, collecting data from such industries to one central server for ML introduces many privacy concerns~\cite{arachchige2019local}. Federated learning (FL) is a recently developed distributed machine learning approach that provides an effective solution to privacy-preserving ML~\cite{li2020federated}. FL lets clients (participants) collect and process data to train a local ML model. The clients are then only required to share the model parameters of the locally trained ML models with a central server for parameter aggregation to generate a global representation of all client models. Finally, the server shares the global model with all participating clients. In this way, FL bypasses the necessity of sharing raw data with any other party involved in the ML training process. However, the model parameters of the locally shared models can still leak private information under certain conditions~\cite{wei2020federated}. Hence, FL on sensitive data such as biometric images, health records, and financial records still poses privacy risks if proper privacy-preservation mechanisms are not imposed.
	
	Cryptographic scenarios and noise addition (randomization) mechanisms have been developed to mitigate the privacy leaks associated with  FL~\cite{zhang2020privcoll,wei2020federated}. Two of FL's most frequently tested cryptographic approaches are secure multi-party computation (SMC) and homomorphic encryption. However, cryptographic approaches tend to reduce FL performance drastically due to their high computational and communication cost~\cite{zhang2020privcoll, bonawitz2017practical}. Most cryptographic approaches assume semi-honest (honest but curious) computations at specific points of the FL process. A semi-honest entity is assumed to conduct computations honestly; however, curious to learn as much information possible~\cite{fereidooni2021safelearn}. Among noise addition approaches, differentially private approaches are more preferred due to the robust privacy guarantees and high efficiency~\cite{wei2020federated,geyer2017differentially}. In global differential privacy (GDP), a trusted curator applies calibrated noise ~\cite{xiao2008output, kairouz2014extremal}, whereas, in local differential privacy (LDP), the data owner perturbs their data before releasing them to any third party~\cite{kairouz2014extremal}. Hence, LDP provides higher levels of privacy as it imposes more noise compared to GDP~\cite{arachchige2019local}. Most existing approaches for FL are based on GDP~\cite{geyer2017differentially}. However, the requirement of a trusted party makes GDP approaches less practical, whereas LDP approaches provide a more practical mode of dealing with the distributed clients in FL. Previous approaches try to impose LDP by applying noise/randomization over the model parameters of the local models~\cite{sun2020ldp}. However, most of these LDP approaches for FL cannot control the privacy budgets efficiently due to the extreme dimensionality of parameter matrices of underlying deep learning models~\cite{sun2020ldp}. For example, the LDP approach for FL proposed in ~\cite{wei2020federated} utilizes extensive $\varepsilon$ values to enable sufficient utility challenging its use in practical settings. Besides, existing LDP approaches are not rigorously tested against more complex datasets~\cite{wei2020federated,truex2020ldp}. Moreover, the weight distribution in different layers of the models has not been explicitly considered during the application of LDP~\cite{sun2020ldp}.
	
	We propose a novel local differentially private federated learning approach (named LDPFL: Local Differential Privacy for Federated Learning) for cross-silo settings that alleviate the issues of existing approaches. LDPFL solves the complexity of applying LDP over high dimensional parameter matrices by applying randomization over a 1D vector generated from the intermediate output of a locally trained model. LDPFL applies randomization over this 1D vector and trains a second fully connected deep neural network as the client model of an FL setting. The randomization mechanism in LDPFL utilizes randomized response (RR)~\cite{fox2015randomized} and optimized unary encoding (OUE)~\cite{wang2017locally} to guarantee differential privacy of FL. Since LDPFL randomizes the inputs to the local models rather than altering the weights of the local models, LDPFL can provide better flexibility in choosing randomization, privacy composition, and model convergence compared to existing LDP approaches for FL. Compared to previous approaches (SMC~\cite{bonawitz2017practical,truex2020ldp}, GDP - DPSGD)~\cite{abadi2016deep,truex2020ldp},   and LDP - $\alpha$-CLDP-Fed~\cite{truex2020ldp,gursoy2019secure}), our empirical analysis shows that  LDPFL performs better and achieves accuracy up
	to 98\% under extreme cases of privacy budgets (e.g., $\varepsilon=0.5$), ensuring a minimal privacy leak.
	
	\section{Background}
	\label{background}
	This section provides brief descriptions of the preliminaries used in LDPFL that were proposed for privacy-preserving federated learning on deep learning in a cross-silo setting. LDPFL utilizes the concepts of local differential privacy (LDP), Randomized Aggregatable Privacy-Preserving Ordinal Response - RAPPOR (an LDP protocol based on randomized response for binary vectors), and optimized unary encoding (which is an optimization on RAPPOR for better utility). 
	
	\subsection{Federated Learning}
	\label{fedlearning}
	
	Federated learning (FL)~\cite{mcmahan2016federated} involves $N$ distributed parties (connected to a central server) agreed on training local deep neural network (DNN) models with the same configuration. The process starts with the central server randomly initializing the model parameters, $\mathcal{M}_0$, and distributing them to the clients to initialize their copy of the model. The clients train their local model separately using the data in their local repository for several local epochs and share the updated model parameters, $M_u$, with the server. The server aggregates the model parameters received from all clients using an aggregation protocol such as federated averaging to generate the federated model $(\mathcal{ML}_{fed})$. Equation~\ref{fedavg} shows the process of federated averaging (calculating the average of values in each index of parameter matrices) to generate $(\mathcal{ML}_{fed})$, where $\mathcal{M}_{u, i}$ represents the updated model parameters sent by $i^{th}$ client. This is called one federation round. FL conducts multiple federation rounds until $(\mathcal{ML}_{fed})$ converges or the pre-defined number of rounds is reached. It was shown that $(\mathcal{ML}_{fed})$ produces accuracy as almost as close to a model centrally trained with the same data~\cite{yang2019federated}.
	
	\begin{equation}
		\label{fedavg}
		\mathcal{ML}_{fed}=\frac{1}{N} \sum_{i} \mathcal{M}_{u, i}
	\end{equation}
	
	\subsection{Local Differential Privacy}
	
	Local differential privacy (LDP) is the setting where the data owners apply randomization (or noise) on the input data before the data curator gains access to them. LDP provides a better privacy notion compared to GDP due to the increased noise levels and nonnecessity of a trusted curator. LDP is deemed to be the state-of-the-art approach for privacy-preserving data collection and distribution. A randomized algorithm $\mathcal{A}$ provides $\varepsilon$-local differential privacy if Equation \eqref{ldpeq} holds~\cite{erlingsson2014rappor}.
	
	\begin{mydef}
		A randomized algorithm $\mathcal{A}$  satisfies $\varepsilon$-local differential privacy if for all pairs of client's values $v_1$ and $v_2$ and for all $Q \subseteq Range(\mathcal{A})$ and for ($\varepsilon \geq 0$), Equation \eqref{ldpeq} holds. $Range(\mathcal{A})$ is the set of all possible outputs of the randomized algorithm $A$.
	\end{mydef}
	\begin{equation}
		Pr[\mathcal{A}(v_1) \in Q] \leq \exp(\varepsilon) Pr[\mathcal{A}(v_2) \in Q]
		\label{ldpeq}
	\end{equation}
	
	\subsection{Randomized Aggregatable Privacy-Preserving
		Ordinal Response (RAPPOR)}
	
	RAPPOR is an LDP algorithm proposed by Google based on the problem of estimating a client-side distribution of string values drawn from a discrete data dictionary~\cite{erlingsson2014rappor}. Basic RAPPOR takes an input, $x_i (\in \mathbb{N}^d)$, that is encoded into a binary string, $\boldsymbol{B}$ of $d$ bits. Each $d$-bit vector contains $d-1$ zeros with one bit at position $v$ set to $1$. Next, $\boldsymbol{B}$ is randomized to obtain $\boldsymbol{B'}$ satisfying DP.
	
	\subsubsection{Sensitivity}
	The sensitivity, $\Delta f$ of a function, $f$, is considered to be the maximum influence that a single individual can have on $f$. In the LDP setting, which involves encoding, this can be represented as given in Equation \eqref{seneq}, where $x_i$ and $x_{i+1}$ are two adjacent inputs, and $f$ represents the encoding. $\lVert . \rVert_1$ represents the $L1$ norm of a vector~\cite{wang2016using}. $\Delta f$ in RAPPOR is 2 as the maximum difference between two adjacent encoded bit strings ($\boldsymbol{B}(x_i)$ and $\boldsymbol{B}(x_{i+1})$) is only two bits. 
	\begin{equation}
		\Delta f=max\{\lVert f(x_i)-f(x_{i+1}) \rVert_1\}
		\label{seneq}
	\end{equation}
	
	\subsubsection{Bit randomization probability}
	\label{rappor}
	Take $p$ to be the probability of preserving the actual value of an original bit in an input bit-string. $p$ follows Equation \eqref{rapeq}, where $\varepsilon$ is the privacy budget offered by the LDP process, as proven by RAPPOR~\cite{erlingsson2014rappor,qin2016heavy}. 
	
	\begin{equation}
		p=\frac{e^{\frac{\varepsilon}{\Delta f}}}{1+e^\frac{\varepsilon}{\Delta f}}=\frac{e^{\frac{\varepsilon}{2}}}{1+e^\frac{\varepsilon}{2}}
		\label{rapeq}
	\end{equation}
	
	\subsection{Optimized Unary Encoding}
	
	Assume that the binary encoding used in RAPPOR (also referred to as Unary Encoding~\cite{wang2017locally}) encodes an input instance $x_i$ into its binary representation $\textbf{B}$, which is a $d$ bit binary vector. Let $\textbf{B}[i]$ be the $i^{th}$ bit and $\textbf{B}^{\prime}[i]$ is the perturbed $i^{th}$ bit. Assume that, one bit at position $v$ of $\textbf{B}$ is set to 1, whereas the other bits are set to zero. Unary Encoding (UE)~\cite{erlingsson2014rappor} perturbs the bits of $\textbf{B}$ according to Equation~\ref{pereq11}. 
	
	\begin{equation}
		\operatorname{Pr}\left[\textbf{B}^{\prime}[i]=1\right]=\left\{\begin{array}{ll}{p,} & {\text { if } \textbf{B}[i]=1} \\ {q,} & {\text { if } \textbf{B}[i]=0}\end{array}\right.
		\label{pereq11}
	\end{equation}
	
	UE satisfies $\varepsilon$-LDP~\cite{erlingsson2014rappor,wang2017locally} for,
	
	\begin{equation}
		\varepsilon=\ln \left(\frac{p(1-q)}{(1-p) q}\right)
		\label{pqratio}
	\end{equation}
	
	This can be proven as done in~\cite{erlingsson2014rappor,wang2017locally} for any bit positions, $v_1$, $v_2$ (of the encoded inputs, $x_1$ and $x_2$, respectively), and output $\textbf{B}$ with sensitivity = 2.

	\begin{proof}
		Considering a sensitivity of 2, choose $p$ and $q$ as follows,
		
		\begin{equation}
			p = \frac{ e^{\frac{\varepsilon}{2}}}{1+ e^{\frac{\varepsilon}{2}}} 
			\label{prob1aa}
		\end{equation}
		\begin{equation}
			q = \frac{1}{1+ e^{\frac{\varepsilon}{2}}} 
			\label{prob2aa}
		\end{equation}
		
		\begin{equation}
			\resizebox{.5\hsize}{!}{$
				\begin{aligned} \frac{\operatorname{Pr}\left[\boldsymbol{B} | v_{1}\right]}{\operatorname{Pr}\left[\boldsymbol{B} | v_{2}\right]} &=\frac{\prod_{i \in[d]} \operatorname{Pr}\left[\boldsymbol{B}[i] | v_{1}\right]}{\prod_{i \in[d]} \operatorname{Pr}\left[\boldsymbol{B}[i] | v_{2}\right]} \\ & \leq \frac{\operatorname{Pr}\left[\boldsymbol{B}\left[v_{1}\right]=1 | v_{1}\right] \operatorname{Pr}\left[\boldsymbol{B}\left[v_{2}\right]=0 | v_{1}\right]}{\operatorname{Pr}\left[\boldsymbol{B}\left[v_{1}\right]=1 | v_{2}\right] \operatorname{Pr}\left[\boldsymbol{B}\left[v_{2}\right]=0 | v_{2}\right]} \\ &=\frac{p}{q} \cdot \frac{1-q}{1-p}=e^{\varepsilon} \end{aligned}
				$}
			\label{moueiproof}
		\end{equation}
		
		Each bit (in a $d$-bit vector) is flipped independently. Equation~\ref{moueiproof}, represents the state where any inputs, $x_1$ and $x_2$ result in bit-vectors that differ only in bit positions $v_1$ and $v_2$. The maximum of this ratio is when $v_1$ is 1 and $v_2$ is 0 as represented by Equation~\ref{moueiproof}.
		
	\end{proof}
	
	Optimized Unary Encoding (OUE) introduces a utility enhancement to Unary Encoding by perturbing 0s and 1s differently. When $\textbf{B}$ is a long binary vector, the number of 0s is significantly greater than the number of 1s in $\textbf{B}$. OUE introduces a mechanism to reduce the probability of perturbing 0 to 1 ($p_{0\rightarrow 1}$). By setting $p=\frac{1}{2}$ and $q=\frac{1}{1+e^\varepsilon}$, OUE improves the budget allocation for transmitting the 0 bits in their original state as much as possible. Following Equation~\ref{moueiproof}, OUE provides $\varepsilon$-LDP when $p=\frac{1}{2}$, $q=\frac{1}{1+e^\varepsilon}$, and  sensitivity = 2~\cite{wang2017locally}.
	
	\subsection{Postprocessing invariance/robustness and composition}
	\label{propdp}
	
	Any additional computations on the outcomes of a DP algorithm do not weaken its privacy guarantees. This property is called the postprocessing invariance/robustness in DP. A processed outcome of a $\varepsilon$-DP algorithm still provides $\varepsilon$-DP. Composition is another property of DP that captures the degradation of privacy when multiple differentially private algorithms are performed on the same or overlapping datasets~\cite{bun2016concentrated}. When two DP algorithms $\varepsilon_1$-DP and $\varepsilon_2$-DP are applied to the same or overlapping datasets, the union of the results is equal to $(\varepsilon_1+\varepsilon_2)$-DP~\cite{bun2016concentrated}. In parallel composition, if a set of DP algorithms ($M_1,M_2,\dots, M_n$) are applied on a dataset $D$ divided into disjoint subsets of $D_1, D_2,\dots, D_n$, respectively (so that $M_i$ provides $\varepsilon_i-DP$ for every $D_i$), the whole process will provide $max\{\varepsilon_1, \varepsilon_2, \dots, \varepsilon_n\}-DP$ on the entire dataset~\cite{zhao2019differential}.
	
	\section{Our Approach}
	\label{approach}
	The proposed approach (to solve the issues raised in Section~\ref{introsection}) is abbreviated as LDPFL (Local Differential Privacy for Federated Learning). Fig.~\ref{ldpflarchitecture} shows the architecture of LDPFL. As shown in Fig.~\ref{ldpflflow}, a client in LDPFL has three main tasks; (1) Generating a fully trained CNN using the local private data (refer to step 1 in Fig.~\ref{ldpflflow}), (2) Generating flattened 1-D vectors of inputs and randomizing them to enforce DP (refer to step 2 and 3 in Fig.~\ref{ldpflflow}),  and (3) Conducting federated learning over randomized data (refer to step 4 in Fig.~\ref{ldpflflow}). In the proposed setting, we assume the clients to be large-scale entities such as banks and hospitals (cross-silo setting), and any data owner would share private data with only one client in the distributed setting (i.e., input data instances are independent). Each client has a private local learning setup where fully trained models are maintained on locally-owned private data. To generalize the models, the clients collaborate with other clients (e.g., hospitals with other hospitals working on similar domains of data) through LDPFL. Each client uses their locally trained CNNs to obtain flattened vectors of the input instances, which are then encoded to binary vectors and randomized to enforce DP. The randomized inputs are then used to train a global model (GM) using federated learning. The following sections provide detailed descriptions of the overall process of LDPFL.
	
	\vspace{-0.5cm}
	\begin{figure}[H] 
		\centering
		\scalebox{0.9}{
			\subfloat[The architecture of LDPFL. \textbf{DO}: data owner, \textbf{CNN}: convolutional neural network \textbf{CM}: convolutional module of the CNN, \textbf{FN}: fully connected network module of the CNN, \textbf{FLT}: input flattening layer, \textbf{RND}: randomization layer, \textbf{DNN}: deep neural network, \textbf{GM}: global model.]{\includegraphics[width=0.70\textwidth, trim=0cm 0cm 0cm 0cm]{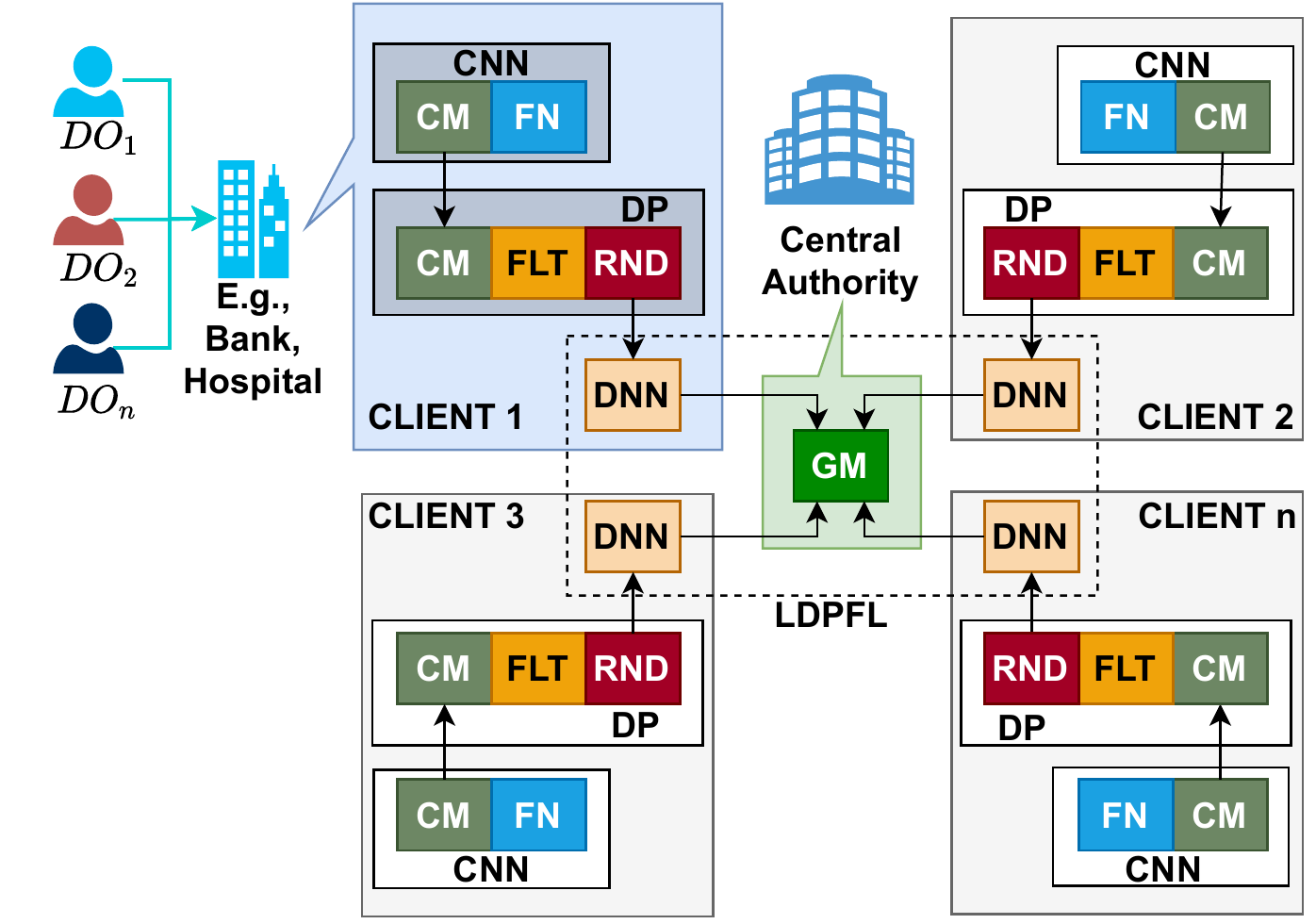}\label{ldpflarchitecture}}
			\vspace{10mm}
			\subfloat[The flow of main steps in  LDPFL.]{\includegraphics[width=0.35\textwidth, trim=0cm 0cm 0cm 0cm]{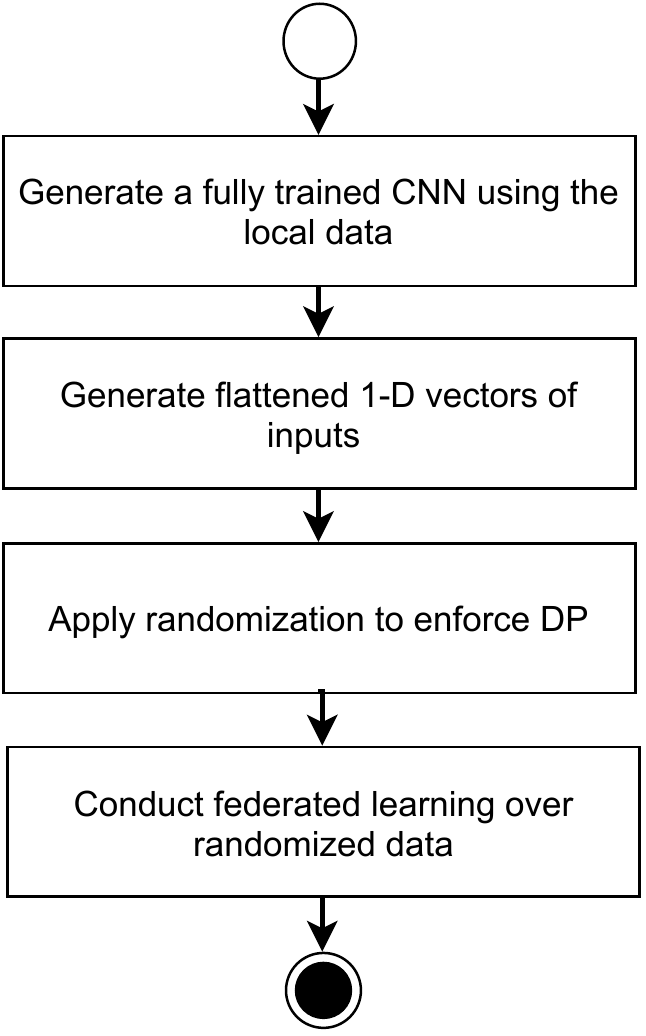}\label{ldpflflow}}
		}
		\caption{LDPFL architecture and its flow of main steps}
		\label{ldplarchiflow}
	\end{figure}

	\subsection{Generating a fully trained CNN using the local private data}
	\label{gencnn}
	The clients use the trained local CNN models to generate the flattened 1-D vectors of the inputs before input encoding and randomization. The randomized input vectors need to be of the same size for the FL setup. Besides, all clients must use the same CNN configurations for the input instances to be filtered through the same architecture of trained convolutional layers to allow a uniform feature extraction procedure.
	
	\subsection{Generating flattened 1-D vectors of inputs and randomizing them to enforce DP}
	\label{genflat}
	Once the CNN models converge on the locally available datasets, the clients use the Convolution module (refer to \textbf{CM} in Fig.~\ref{ldpflarchitecture}) of the converged CNN models to predict 1-D flattened outputs from the last convolutional layer of the \textbf{CM} for all inputs. Next, the predicted flattened outputs (1-D vectors: $1DV$) are encoded to binary vectors, which are then randomized to produce DP binary vectors. Utilizing a fully trained client CNN model for data flattening enables LDPFL to preserve the representative features of the input data and maintain the attribute distributions to generate high utility.
	
	\subsubsection{Binary encoding}
	
	Each element of a $1DV$ is converted to a binary value (binary representation) according to Equation \eqref{intbinary}. $m$ and $n$ are the numbers of binary digits of the whole number and the fraction, respectively. $x$ represents the original input value where $x\in \mathbb{R}$, and  $g(i)$ represents the $i^{th}$ bit of the binary string where the least significant bit is represented when $k=-m$. Positive numbers are represented with a sign bit of 0, and negative numbers are represented with a sign bit of 1.
	\begin{equation}
		\resizebox{.6\hsize}{!}{$
			g(i)=\left(\left\lfloor 2^{-k} x\right\rfloor \bmod 2\right)_{k=-m}^{n} \quad \text { where } i=k+m
			$}
		\label{intbinary}
	\end{equation}
	The binary conversion's sensitivity and precision (the range of floating values represented by the binary numbers) can be changed by increasing or decreasing the values chosen for $n$ and $m$. Separately randomizing each binary value adds up the privacy budget after each randomization step according to the composition property of DP (refer to Section~\ref{propdp}). Besides, dividing the privacy budget among the binary values introduces unreliable levels of bit randomization. Hence, we merge all binary values into one long binary vector ($L_{b}$) before the randomization to consume the privacy budget of randomization efficiently. Large values for $n$ and $m$ can result in undesirably long binary vectors for randomization. Hence, $n$ and $m$ must be chosen carefully by empirically evaluating and adjusting them to produce high model accuracy.
	
	\subsubsection{Randomization}
	The length of an encoded binary string is $l=(m+n+1)$; hence, the full length of a merged binary string ($L_{b}$) is $rl$ (take, $v$ to be any bit position of $L_{b}$ set to 1), where $r$ is the total number of outputs of the $1DV$. Consequently, the sensitivity of the encoded binary strings can be taken as $l\times r$, as two consecutive inputs can differ by at most $l\times r$ bits. Now the probability of randomization can be given by Equation \eqref{mergedrandprob} (according to Equation \ref{rapeq}).
	
	\begin{equation}
		p=\frac{e^{\varepsilon/rl}}{1+e^{\varepsilon/rl}}
		\label{mergedrandprob}
	\end{equation}
	
	With $p$ probability of randomization, the probability of randomization in reporting opposite of the true bits is $(1-p) = \frac{1}{1+e^{\varepsilon/rl}}$. This probability can lead to an undesirable level of randomization (with UE or OUE) due to the extremely high sensitivity $rl$. Hence, LDPFL employs an optimized randomization mechanism that further optimizes OUE to perturb 0s and 1s differently, reducing the probability of perturbing 0 to 1. In this way, LDPFL tries to maintain the utility at a high level under the high sensitivity of concatenated binary vectors, $L_{b}$s. The parameter, $\alpha$ (the privacy budget coefficient) is introduced as defined in Theorem~\ref{mouelarge} to improve the flexibility of randomization probability selection further while still guaranteeing $\varepsilon$-$LDP$. By increasing $\alpha$, we can increase the probability of transmitting the 0 bits in their original state. 
	
	\begin{theorem}
		\label{mouelarge}
		Let $v_1$, $v_2$ be any equally distributed bit positions of any two binary vectors $L_{b_1}$ and $L_{b_2}$, respectively, and $\textbf{B}$ be a $d$-bit binary string output. When $\operatorname{Pr}\left[\boldsymbol{B}\left[v_{1}\right]=1 | v_{1}\right] = \frac{1}{1+\alpha}$, 
		%
		%
		$\operatorname{Pr}\left[\boldsymbol{B}\left[v_{2}\right]=0 | v_{1}\right] = \frac{\alpha e^{\frac{\varepsilon}{{rl}/2}}}{1+\alpha e^{\frac{\varepsilon}{{rl}/2}}} $,
		%
		%
		the randomization provides $\varepsilon$-$LDP$.
	\end{theorem}
	
	\begin{proof}
		Let $\varepsilon$ be the privacy budget and $\alpha$ be the privacy budget coefficient. 
		
		\begin{equation}
			\resizebox{.5\hsize}{!}{$
				\begin{aligned} \frac{\operatorname{Pr}\left[\boldsymbol{B} | v_{1}\right]}{\operatorname{Pr}\left[\boldsymbol{B} | v_{2}\right]} &=\frac{\prod_{i \in[d]} \operatorname{Pr}\left[\boldsymbol{B}[i] | v_{1}\right]}{\prod_{i \in[d]} \operatorname{Pr}\left[\boldsymbol{B}[i] | v_{2}\right]} \\ & \leq \left( \frac{\operatorname{Pr}\left[\boldsymbol{B}\left[v_{1}\right]=1 | v_{1}\right] \operatorname{Pr}\left[\boldsymbol{B}\left[v_{2}\right]=0 | v_{1}\right]}{\operatorname{Pr}\left[\boldsymbol{B}\left[v_{1}\right]=1 | v_{2}\right] \operatorname{Pr}\left[\boldsymbol{B}\left[v_{2}\right]=0 | v_{2}\right]} \right)^{{rl}/2} \\ &=\left(\frac{\left(\frac{1}{1+\alpha}\right)}{\left(\frac{\alpha}{1+\alpha}\right)} \cdot \frac{\left(\frac{\alpha e^{\frac{\varepsilon}{{rl}/2}}}{1+\alpha e^{\frac{\varepsilon}{{rl}/2}}} \right)}{\left(\frac{1}{1+\alpha e^{\frac{\varepsilon}{{rl}/2}}} \right)}\right)^{{rl}/2}=e^{\varepsilon} \end{aligned}
				$}
			\label{repproof}
		\end{equation}
		
	\end{proof}
	
	Theorem \ref{mouelarge} provides the flexibility for selecting the randomization probabilities at large $\alpha$ values. However, it can introduce undesirable randomization levels on 1s when the bit string is too long (e.g., more than 10,000 bits). Hence, we extend Theorem \ref{mouelarge} further to impose additional flexibility over bit randomization. This is done by employing two randomization models over the bits of $L_{b}$, by randomizing one half of the bit string differently from the other half while still preserving $\varepsilon-LDP$ as defined in Theorem~\ref{uertheorem}. Consequently, Theorem \ref{uertheorem} applies less randomization on ${\frac{3}{4}}^{th}$ of a binary string while other ${\frac{1}{4}}^{th}$ of the binary string is heavily randomized. In this way, the randomization can maintain a high utility for extensively long binary strings as a significant part of the binary string is still preserved.
	
	\begin{theorem}
		\label{uertheorem}
		Let $Pr(\textbf{B}|v)$ be the probability of randomizing a bit for any input bit position $v$ and output $\textbf{B}$. For any equally distributed input bit positions, $v_1, v_2$ of any two binary vectors, $L_{b_1}$ and $L_{b_2}$, respectively,  with a sensitivity = $rl$, define the probability, $Pr(\textbf{B}|v)$ as in Equation~\ref{twomodprob}. Then the randomization provides $\varepsilon$-LDP.  
		
	\end{theorem}

	\begin{proof}
		Choose the randomization probabilities according to Equation~\ref{twomodprob}.
		
		\begin{equation}
			\resizebox{.5\hsize}{!}{$
				Pr(\textbf{B}|v)=\left\{\begin{array}{lc}
					\operatorname{Pr}\left[\boldsymbol{B}\left[v_{1}\right]=1 \mid v_{1}\right]=\frac{\alpha}{1+\alpha} & \text { if } i \in \mathcal{S}_1 \\
					\operatorname{Pr}\left[\boldsymbol{B}\left[v_{2}\right]=0 \mid v_{1}\right]=\frac{\alpha e^{\frac{\varepsilon}{{rl}/2}}}{1+\alpha e^{\frac{\varepsilon}{{rl}/2}}} & \text { if } i \in\mathcal{S}_1\ \\
					\operatorname{Pr}\left[\boldsymbol{B}\left[v_{1}\right]=1 \mid v_{1}\right]=\frac{1}{1+\alpha^{3}} & \text { if } i \in \mathcal{S}_2 \\
					\operatorname{Pr}\left[\boldsymbol{B}\left[v_{2}\right]=0 \mid v_{1}\right]=\frac{\alpha e^{\frac{\varepsilon}{{rl}/2}}}{1+\alpha e^{\frac{\varepsilon}{{rl}/2}}} & \text { if } i \in \mathcal{S}_2
				\end{array}\right.
				$}
			\label{twomodprob}
		\end{equation}
		
		where, $\mathcal{S}_1 = \{2 n \mid n \in \mathbb{N}\}$ and $\mathcal{S}_2 = \{2 n+1 \mid n \in \mathbb{Z}^+\}$.
		
		\begin{equation}
			\resizebox{.6\hsize}{!}{$
				\begin{aligned}
					\frac{\operatorname{Pr}\left[\boldsymbol{B} | v_{1}\right]}{\operatorname{Pr}\left[\boldsymbol{B} | v_{2}\right]} &= \frac{\prod_{i \in[d]} \operatorname{Pr}\left[\boldsymbol{B}[i] | v_{1}\right]}{\prod_{i \in[d]} \operatorname{Pr}\left[\boldsymbol{B}[i] | v_{2}\right]} \\
					&=\frac{\prod_{i \in\mathcal{S}_1} \operatorname{Pr}\left[\boldsymbol{B}[i] | v_{1}\right]}{\prod_{i \in\mathcal{S}_1} \operatorname{Pr}\left[\boldsymbol{B}[i] | v_{2}\right]}\times \frac{\prod_{i \in \mathcal{S}_2} \operatorname{Pr}\left[\boldsymbol{B}[i] | v_{1}\right]}{\prod_{i \in\mathcal{S}_2} \operatorname{Pr}\left[\boldsymbol{B}[i] | v_{2}\right]} \\ 
					& \leq \left( \frac{\operatorname{Pr}\left[\boldsymbol{B}\left[v_{1}\right]=1 | v_{1}\right] \operatorname{Pr}\left[\boldsymbol{B}\left[v_{2}\right]=0 | v_{1}\right]}{\operatorname{Pr}\left[\boldsymbol{B}\left[v_{1}\right]=1 | v_{2}\right] \operatorname{Pr}\left[\boldsymbol{B}\left[v_{2}\right]=0 | v_{2}\right]} \right)^{rl/4}\times \\
					&\left( \frac{\operatorname{Pr}\left[\boldsymbol{B}\left[v_{1}\right]=1 | v_{1}\right] \operatorname{Pr}\left[\boldsymbol{B}\left[v_{2}\right]=0 | v_{1}\right]}{\operatorname{Pr}\left[\boldsymbol{B}\left[v_{1}\right]=1 | v_{2}\right] \operatorname{Pr}\left[\boldsymbol{B}\left[v_{2}\right]=0 | v_{2}\right]} \right)^{rl/4} \\
					&=\left(\frac{\left(\frac{\alpha}{1+\alpha}\right)}{\left(\frac{1}{1+\alpha}\right)} \cdot \frac{\left(\frac{\alpha e^{\frac{\varepsilon}{rl/2}}}{1+\alpha e^{\frac{\varepsilon}{rl/2}}} \right)}{\left(\frac{1}{1+\alpha e^{\frac{\varepsilon}{rl/2}}} \right)}\right)^{rl/4} \left(\frac{\left(\frac{1}{1+\alpha^3}\right)}{\left(\frac{\alpha^3}{1+\alpha^3}\right)} \cdot \frac{\left(\frac{\alpha e^{\frac{\varepsilon}{rl/2}}}{1+\alpha e^{\frac{\varepsilon}{rl/2}}} \right)}{\left(\frac{1}{1+\alpha e^{\frac{\varepsilon}{rl/2}}} \right)}\right)^{rl/4}\\
					&=e^{\varepsilon} \end{aligned}
				$}
		\end{equation}
		
	\end{proof}
	
	\subsection{Conducting federated learning over randomized data}
	\label{randomization}
	After declaring the FL setup, the clients feed the randomized binary vectors as inputs to the FL setup of LDPFL. We assume that all examples are independent and that clients do not collude with one another. As shown in Fig.~\ref{ldpflarchitecture}, after the initialization of the local models, all clients train a local model (as represented by DNN in the figure) using the randomized inputs for a certain number of local epochs and transfer the trained model parameters to the server. Since LDPFL uses local differential privacy at each client and all examples are independent, the final privacy budget consumption is the maximum of all privacy budgets used by each client ($max\{\varepsilon_1, \varepsilon_2, \dots, \varepsilon_n\}-DP$). Algorithm~\ref{ranalgo} shows the composition of the steps (explained in Section \ref{approach}) of LDPFL in conducting differentially private federated learning that satisfies $\varepsilon-LDP$. 
	
	\begin{center}
		\scalebox{0.67}{
			\begin{minipage}{1.1\linewidth}
				\removelatexerror
				\begin{algorithm}[H]
					\caption{LDPFL Algorithm}
					\label{ranalgo}
					\begin{multicols}{2}
						\KwIn{
							\begin{tabular}{l c l} 
								$\{cx_1,\dots, cx_z\}               $ & $\gets $ & client datasets\\
								& & of $z$ clients\\
								$\varepsilon              $ & $\gets $ & privacy budget\\
								$m              $ & $\gets $ & number of bits  \\
								& & for the whole \\
								& &  number\\
								$n              $ & $\gets $ & number of bits \\
								& & for the fraction\\
								$\alpha$ & $\gets$ & privacy budget\\ 
								& & coefficient\\
								$el$ & $\gets$ & the total \\
								& & number of\\ 
								& & local epochs\\
								$E$ & $\gets$ & the total \\
								& & number of \\
								& & global rounds
							\end{tabular}
						}
						
						\KwOut{
							\begin{tabular}{ l c l } 
								$GM$ & $ \gets $ &  differentially private\\
								& &  global model\\
							\end{tabular}
						}
						\textbf{Part I: Randomized data  ~~~~~~generation at clients:}\\
						Declare $i^{th}$ client's model $CNN_i$ for  each ~~~~ client ($C_i$) (refer to  Section~\ref{gencnn})\;
						Train $CNN_i$ with $cx_i$ until the convergence\;  
						Split trained $CNN_i$ into $CM_i$ and $FN_i$ (refer to Section~\ref{genflat} and Fig.~\ref{ldplarchiflow})\;
						Declare, $l=(m+n+1)$\;
						Feed $cx_i$ to $CM_i$ and generate the sequence of 1-D feature arrays $\{d_1,\dots, d_j\}_i$ for $j$ data samples in $i^{th}$ client\;
						Convert  each field ($x$) of $d_q$ (where, $q=1,\dots,j$) to binary using, $g(i)=\,{\Big(\left\lfloor 2^{-k}\, \abs x \right\rfloor\text{ mod }2\Big)_{k=-m}^{n}}\ \text{where, } i=k+m$\;
						Generate arrays $\{FLT_1,\dots, FLT_j\}_i$ by merging the binary arrays  of each $d_q$ ~~~~ in $\{d_1,\dots, d_j\}_i$\;
						Calculate the randomization probability, ~~~~ $p$ according to Equation~\ref{twomodprob}\; 
						Randomize each $FLT_q$ of $\{FLT_1,\dots, FLT_j\}_i$ based on Theorem~\ref{uertheorem} to generate $\{RND_1, \dots, RND_j\}_i$\;
						Declare client models  ($DNN_i$) of each client (and the server - $GM$) for FL\;   
						\textbf{Part II: Federated learning:}\\
						Server randomly initializes model parameters ($M_0$)\;
						Server sends $M_0$ to the $z$ clients\;
						Clients initialize $DNN_i$ using $M_0$\;
						$e = 1$\;
						\While{e $\leq$ E}{
							\For{each client, $C_i$ in the current round}{
								
								Train $DNN_i$ using $\{RND_1, \dots, RND_j\}_i$ for $el$ epochs\;
								Send updated parameters $M_{u_i}$ to the server\;
							}
							Conduct federated averaging,
							$\mathcal{ML}_{fed}=\frac{1}{v} \sum_{i} \mathcal{M}_{u_i}$ (for $v$ clients contributed to the current round)\;
							Update client models ($DNN$s) with $\mathcal{ML}_{fed}$\;
							$e = e + 1$\;
						}
						GM = $\mathcal{ML}_{fed}$\;
						return GM\;
					\end{multicols}
				\end{algorithm}
			\end{minipage}%
		}
	\end{center}
	
	\section{Results and Discussion}
	\label{resdis}
	To test LDPFL, we use the MNIST~\cite{lecun1998gradient}, the CIFAR10~\cite{abadi2016deep}, the SVHN~\cite{sermanet2012convolutional}, and the FMNIST~\cite{xiao2017fashion} datasets. CIFAR10 is a much more complex dataset to be trained than MNIST. Hence, these two datasets introduce a balanced experimental setting for LDPFL performance testing. However, MNIST and CIFAR10 have a limited number of examples of 70,000 and 60,000 images, respectively. Hence, an extensive dataset is necessary to enable all clients to have a large enough dataset partition to test LDPFL's performance under a large number of clients. We use SVHN with 600,000 images to solve this problem. Besides, we use the FMNIST dataset for the performance comparison of LDPFL against previous approaches following the benchmarking conducted in~\cite{truex2020ldp}. We used a MacBook pro-2019 computer for single program experimentations. It has a processing unit of 2.4 GHz 8-Core Intel Core i9 and a memory of 32 GB 2667 MHz DDR4. We used one 112 Dual Xeon 14-core E5-2690 v4 Compute Node (with 256 GB RAM and 4 Tesla P100-SXM2-16GB GPUs)  of the CSIRO Bracewell HPC cluster for multi-round experimentation (repeating the experiments multiple rounds in parallel). We repeated all experiments ten times in the CSIRO Bracewell HPC cluster and reported the average performance to maintain the stability of the results.

	\subsection{LDPFL architectural configurations and datasets used during the experiments}
	\label{ldpflconfig}
	We used two LDPFL architectural configurations to study the performance under different dynamics of the datasets used, as the correct configuration leads to high model quality~\cite{schmidhuber2015deep}. Fig.~\ref{ldpflminst} shows the architecture used for the MNIST dataset. As shown in Fig.~\ref{ldpflcifar10}, we used a comparably complex configuration for CIFAR10, FMNIST, and SVHN as they are more complex datasets compared to MNIST. Figures~\ref{ldpflminst} and~\ref{ldpflcifar10} show the flow of modules in LDPFL, layer types used in the networks,  the input size of each layer, and the layer order from top to bottom. The input size of a particular layer also indicates the output size of the previous layer. The resolution of an image in FMNIST is 28x28x1 was different from the image resolution (32x32x3) in CIFAR10 and SVHN. Hence, we made necessary modifications (discussed in Sections~\ref{cnnconfig} and~\ref{fedlearn}) to the architecture in Fig.~\ref{ldpflcifar10} to accommodate the change in the input size when LDPFL was tested on FMNIST. As shown in Figures~\ref{ldpflminst} and~\ref{ldpflcifar10} we used the Python Keras API~\cite{chollet2015keras} to implement the CNN and DP modules. The federated learning module of LDPFL was implemented using the PyTorch API~\cite{paszke2019pytorch}. 
	\vspace{-0.4cm}
	\subsection{Conducting experiments on LDPFL}
	\subsubsection{Distributing data among clients for the experiments}
	We split the total number of records into groups with equal numbers of records according to the highest number of clients -- $N_h$ (the LDPFL was going to be tested on). Hence, a particular client holds a total of $\frac{T_r}{N_h}$ records, where $T_r$ is the total number of records. However, for the experiments on highly imbalanced data (the non-IID setting), we randomly distributed FMNIST data among ten clients with high sparseness, as shown in Fig.~\ref{FMNISTimbadist}.
	\newpage
	\begin{multicols}{2}
		\begin{figure}[H] 
			\centering
			\includegraphics[width=0.18\textwidth, trim=0cm 0cm 0cm 0cm]{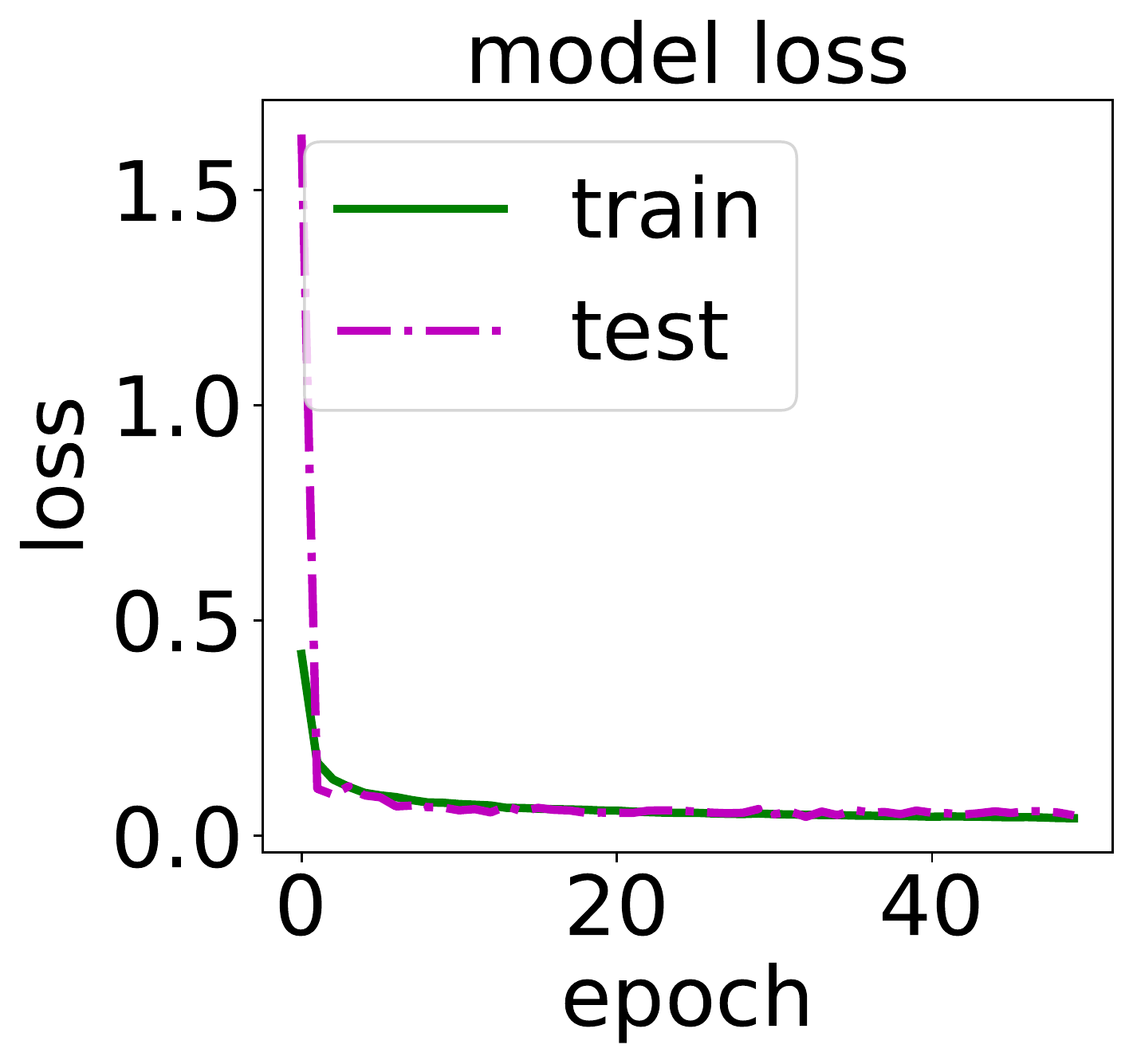}
			\includegraphics[width=0.19\textwidth, trim=0cm 0cm 0cm 0cm]{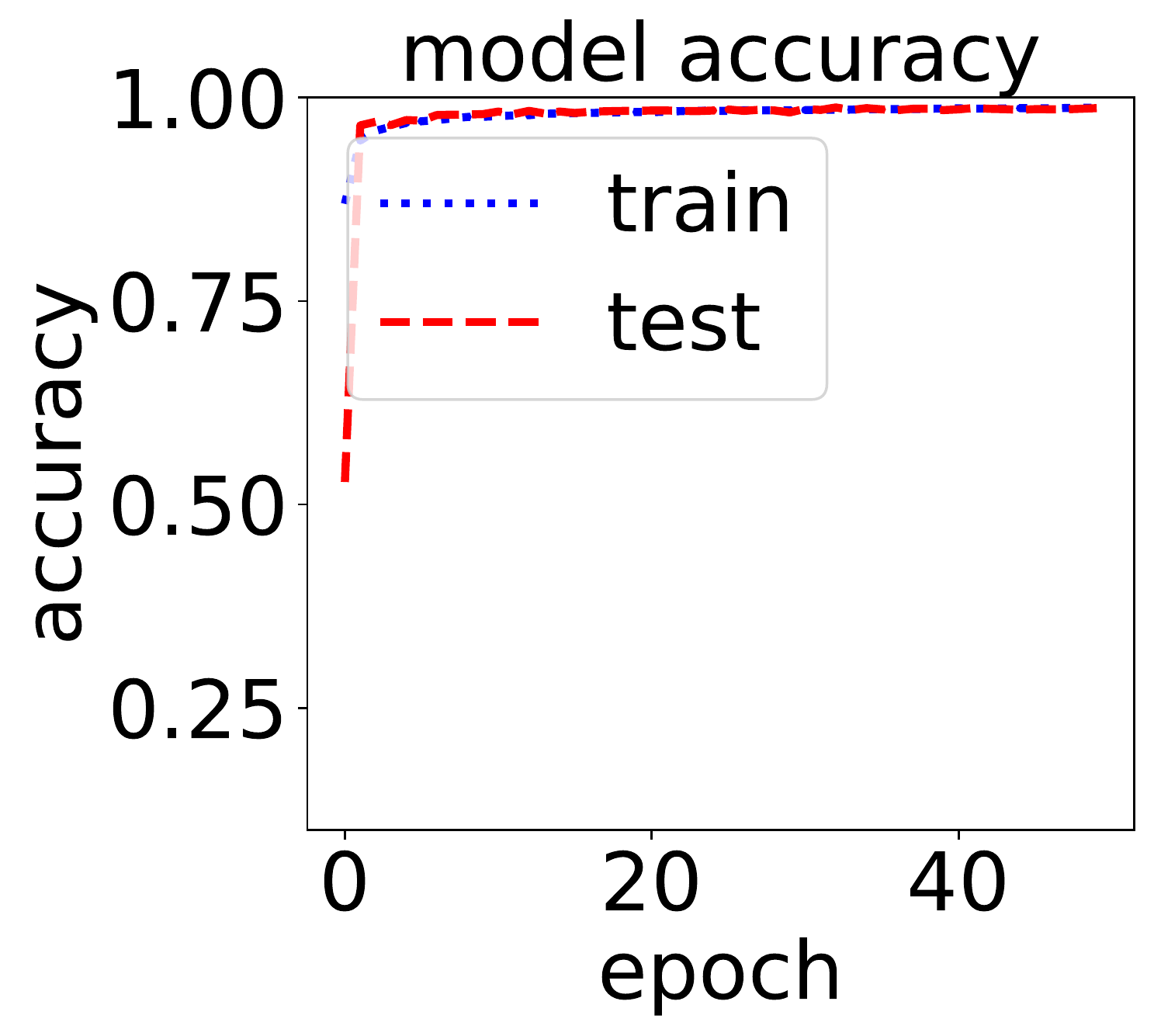}
			\caption{MNIST local model (refer to CNN in Figures~\ref{ldpflarchitecture} and~\ref{ldpflminst}) performance }
			\label{mnistcnnperformance}
		\end{figure}
		\begin{figure}[H] 
			\centering
			\includegraphics[width=0.18\textwidth, trim=0cm 0cm 0cm 0cm]{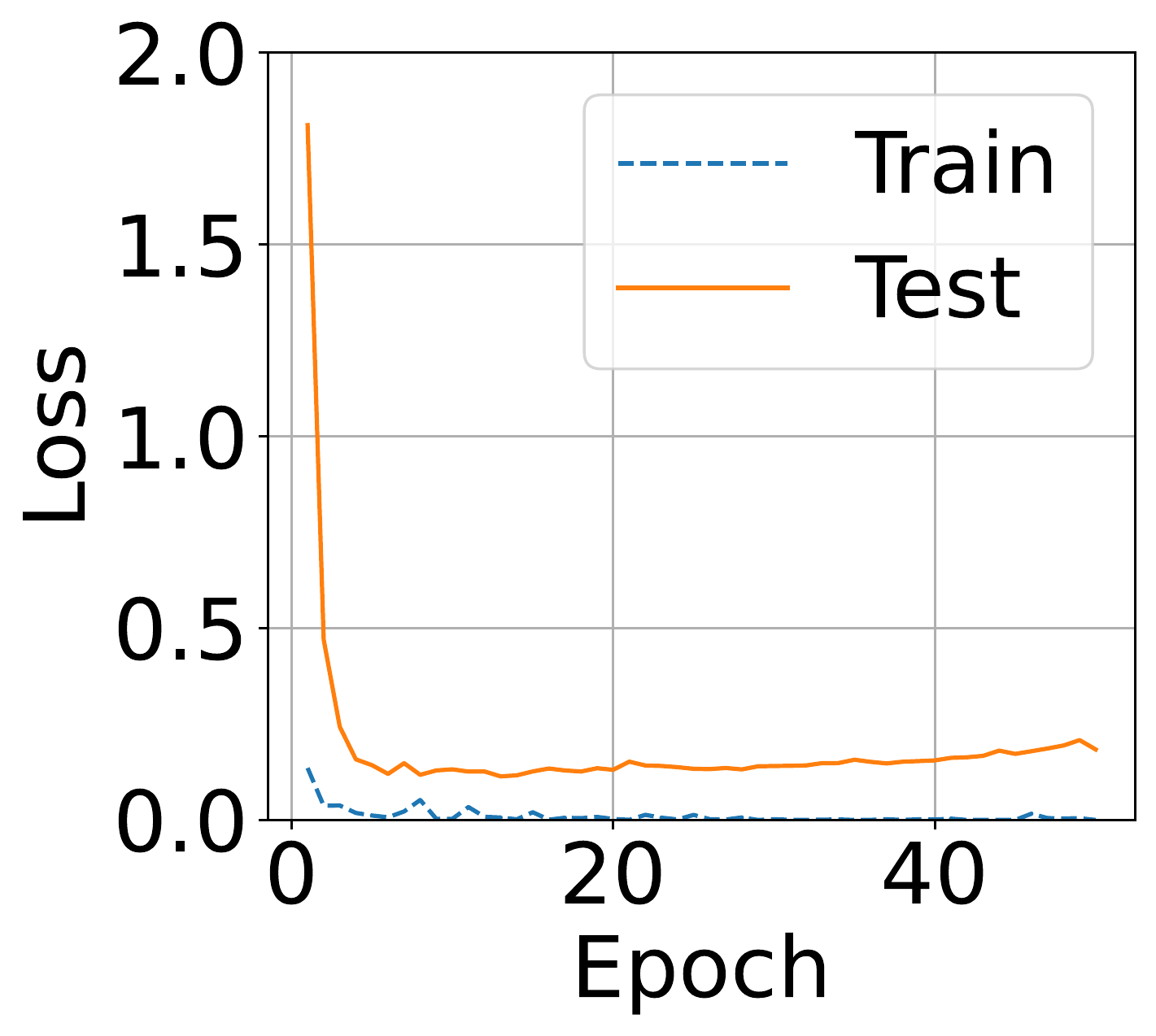}
			\includegraphics[width=0.19\textwidth, trim=0cm 0cm 0cm 0cm]{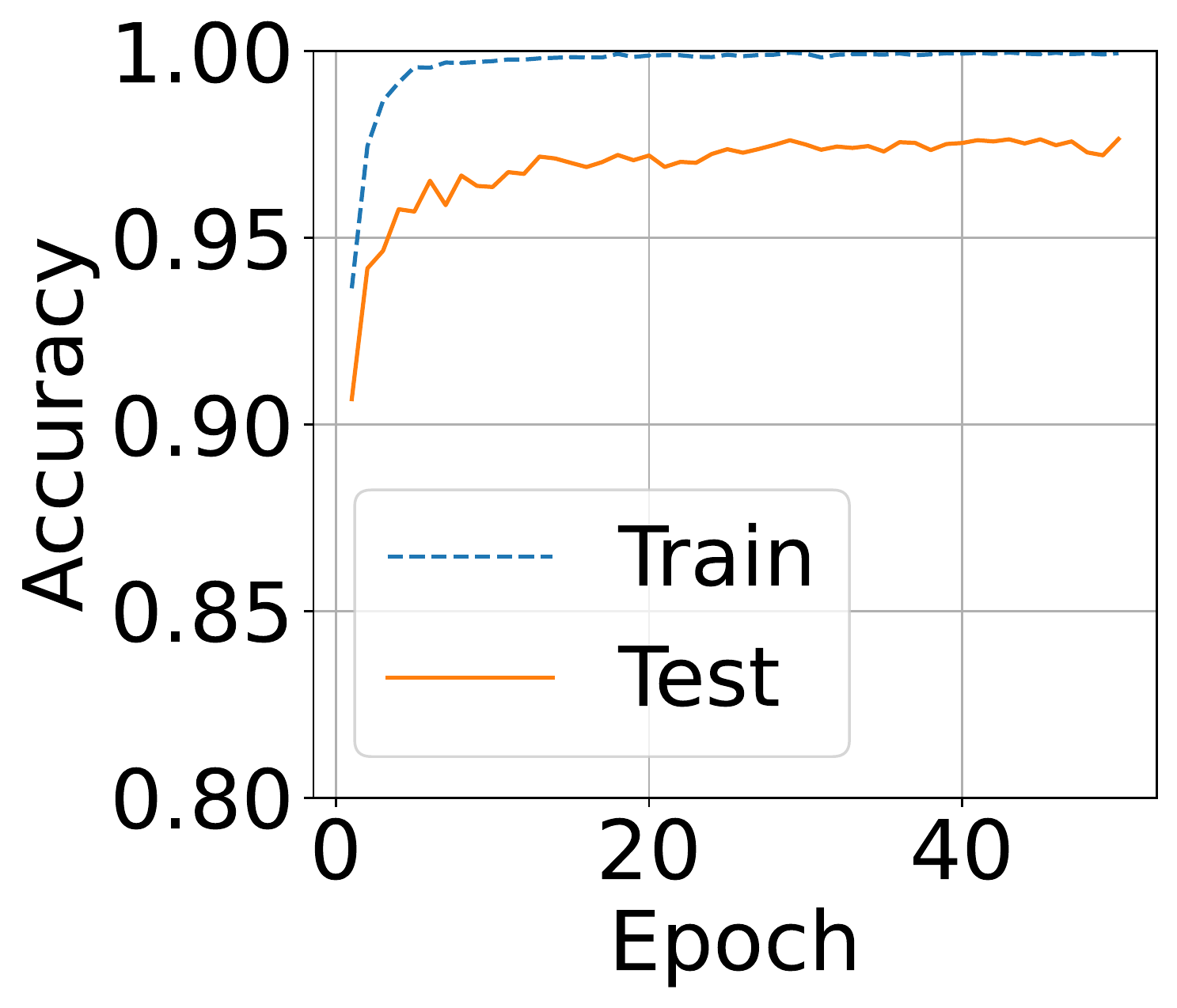}
			\caption{MNIST global model (refer to GM in Figures~\ref{ldpflarchitecture}  and~\ref{ldpflminst}) performance}
			\label{MNISTLDPFLperformance}
		\end{figure}
	\end{multicols}
	
	\vspace{-0.9cm}
	\subsubsection{Training client CNN models with image augmentation}
	We used 60000, 50000, 60000, and  451461 training samples and 10000, 10000, 10000, and 79670 testing samples, and $N_h$ was set to 2, 2, 10, and 100 under MNIST, CIFAR10, FMNIST, and SVHN, respectively. Hence, each client had 30000, 25000, 6000, and 4514 data samples for training, and 5000, 5000, 1000, and 796 testing samples under MNIST, CIFAR10, FMNIST, and SVHN, respectively. Each client used 90\% of local data for training and 10\% for testing. All clients used image augmentation to maintain a high local model performance and robustness under a low number of data samples. We used RMSprop(lr=0.001,decay=1e-6) optimizer for local CNN training with a batch size of 64. All CNNs were trained for 50 epochs. Figures~\ref{mnistcnnperformance} and~\ref{cifar10cnnperformance} show the two CNN client model performances under MNIST and CIFAR10, respectively. Fig.~\ref{fmnistmodelperformance} shows the CNN model performance of a randomly chosen one of the ten clients under FMNIST. From the 100 dataset splits of SVHN, we only considered a maximum of 50 clients as it provided enough evidence to understand the LDPFL performance patterns against the increasing number of clients. Fig.~\ref{svhnmodelperformance} shows the CNN model performance of a randomly chosen one of the 50 clients. The client CNN performance plots (\ref{mnistcnnperformance},~\ref{cifar10cnnperformance},~\ref{fmnistmodelperformance}, and ~\ref{svhnmodelperformance}) show that the configurations chosen for the client CNN models under each dataset generate good model performance. 
	
	\vspace{-0.5cm}
	\begin{minipage}[c]{0.38\linewidth}
		\begin{figure}[H] 
			\centering
			\includegraphics[width=0.44\textwidth, trim=0cm 0cm 0cm 0cm]{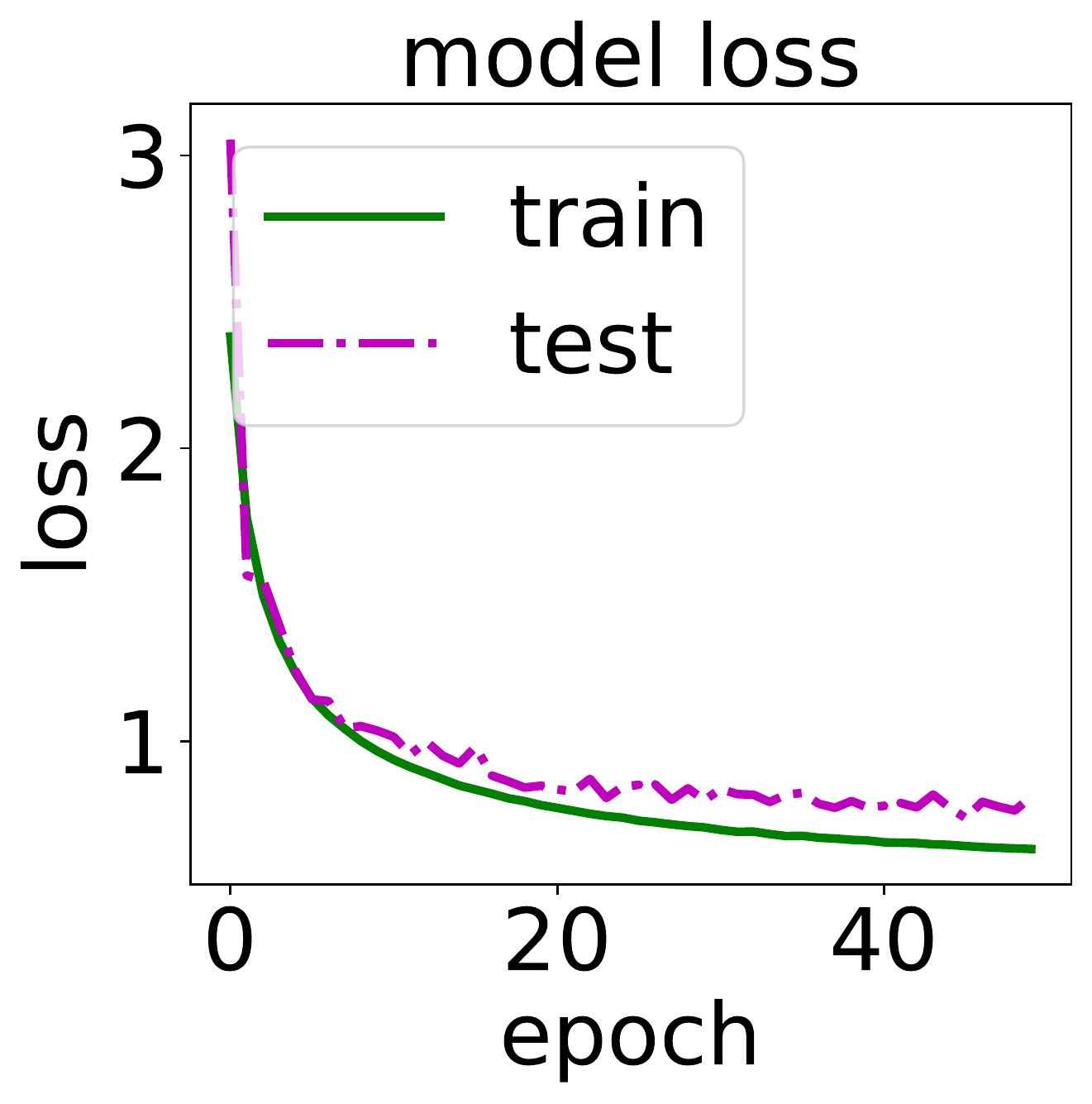}
			\hspace{0.2cm}
			\includegraphics[width=0.47\textwidth, trim=0.3cm 0cm 0cm 0cm]{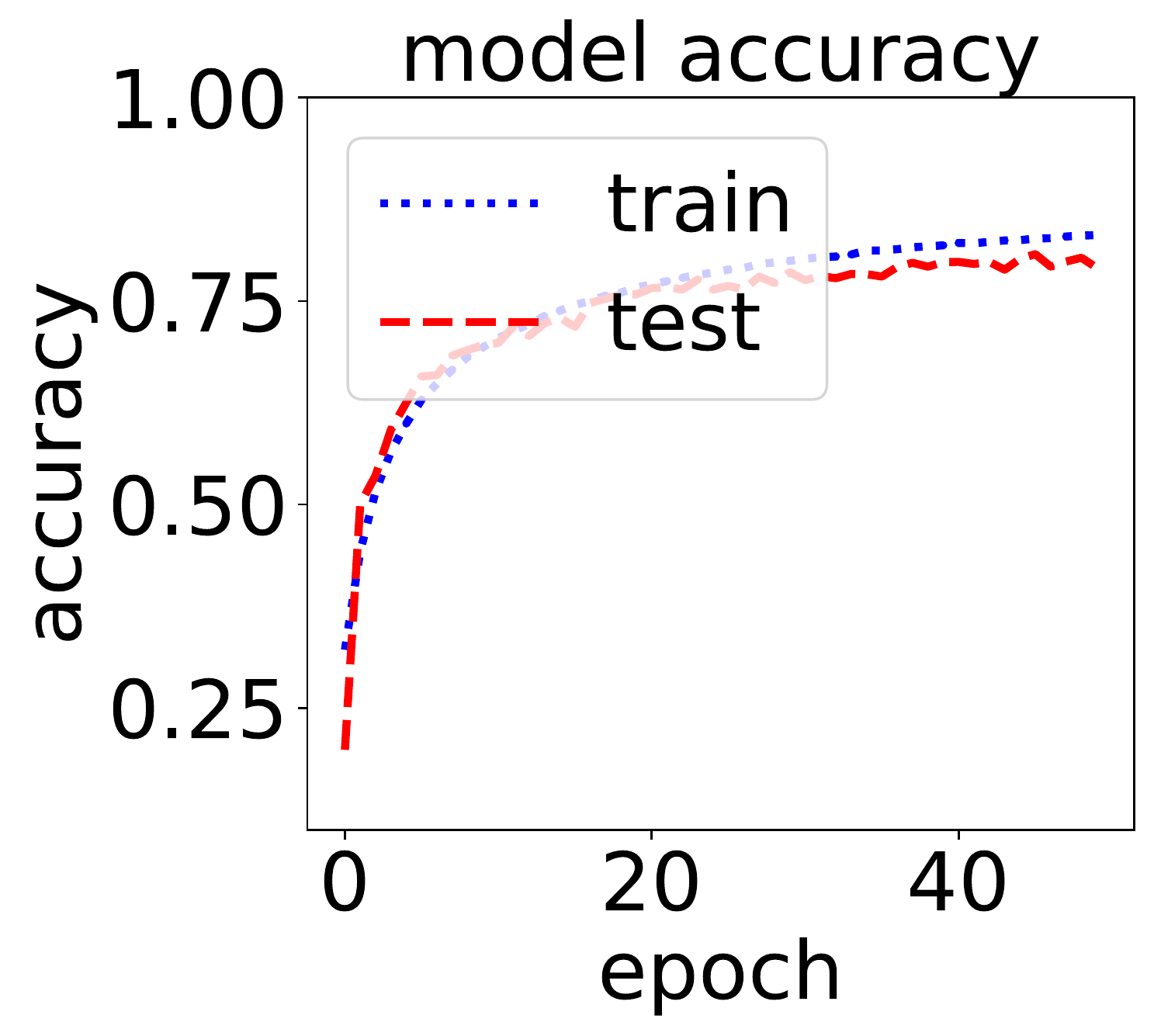}
			\caption{CIFAR10 local model (refer to CNN in Figures~\ref{ldpflarchitecture} and~\ref{ldpflcifar10}) performance}
			\label{cifar10cnnperformance}
		\end{figure}
	\end{minipage}
	\hspace{0.2cm}
	\begin{minipage}[c]{0.54\linewidth}
		\begin{figure}[H] 
			\centering
			\scalebox{1.02}{
				\includegraphics[width=0.28\textwidth, trim=0cm 0cm 0cm 0cm]{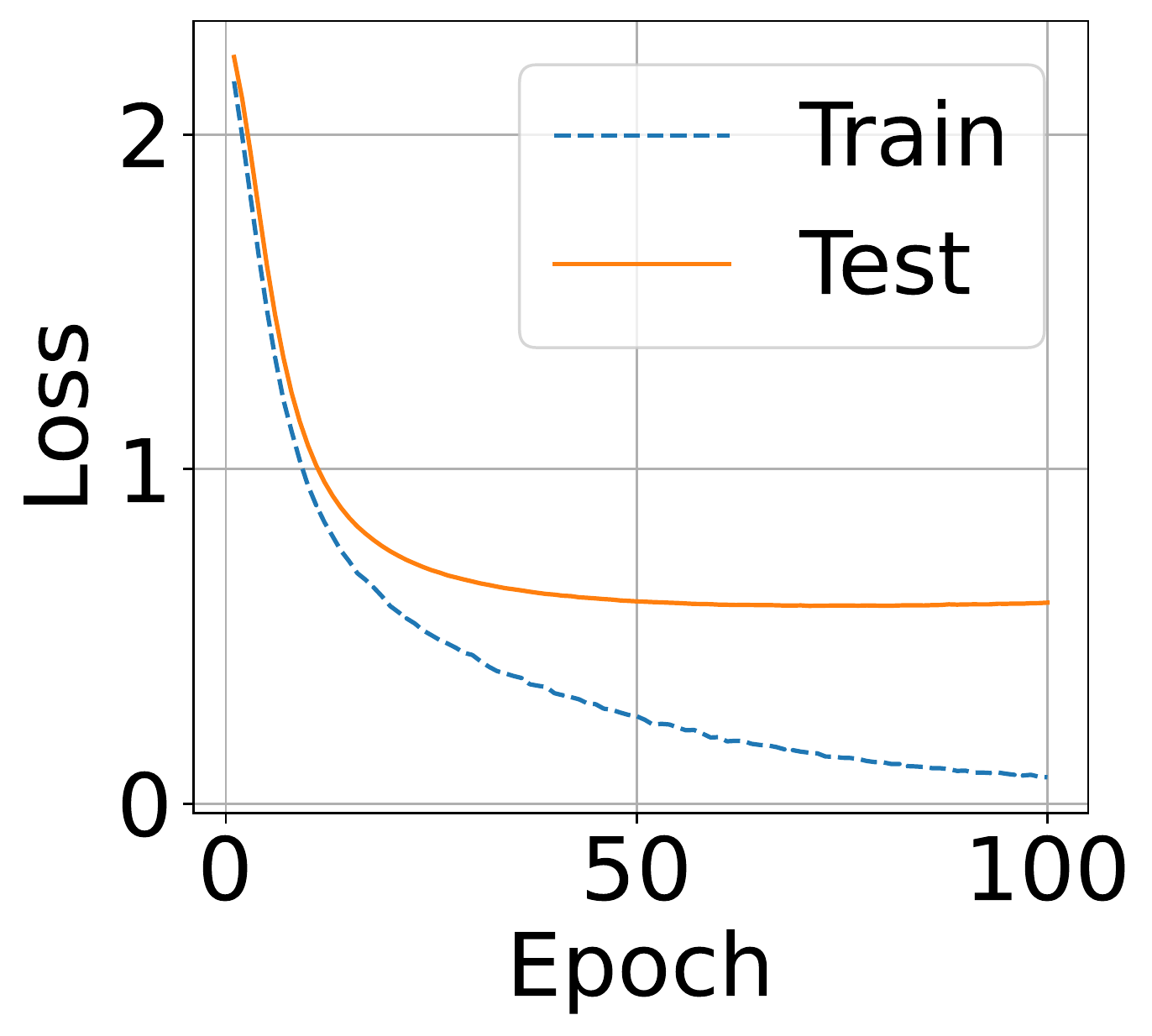}
				\includegraphics[width=0.30\textwidth, trim=0cm 0cm 0cm 0cm]{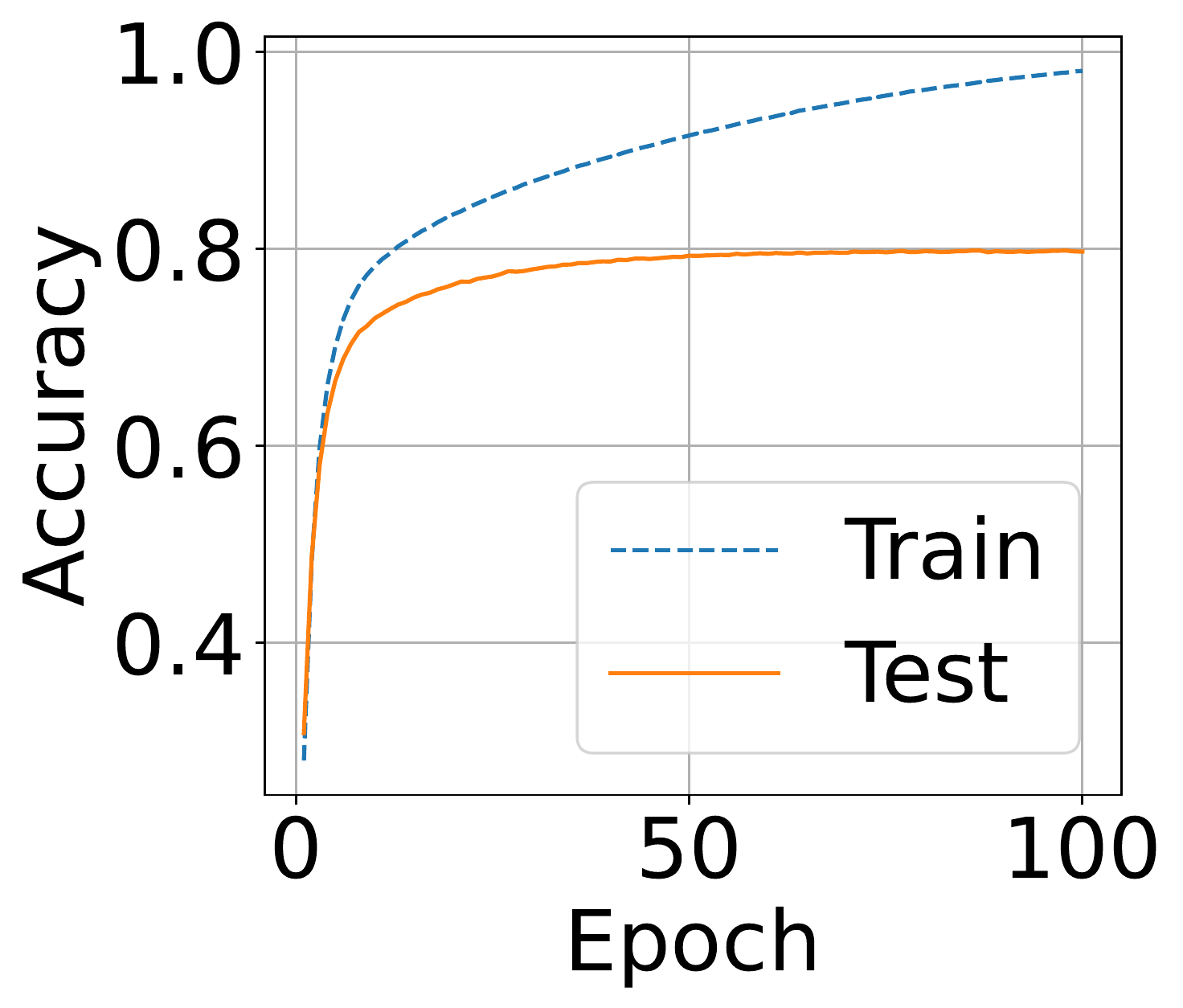}
				\includegraphics[width=0.32\textwidth, trim=0cm 0cm 0cm 0cm]{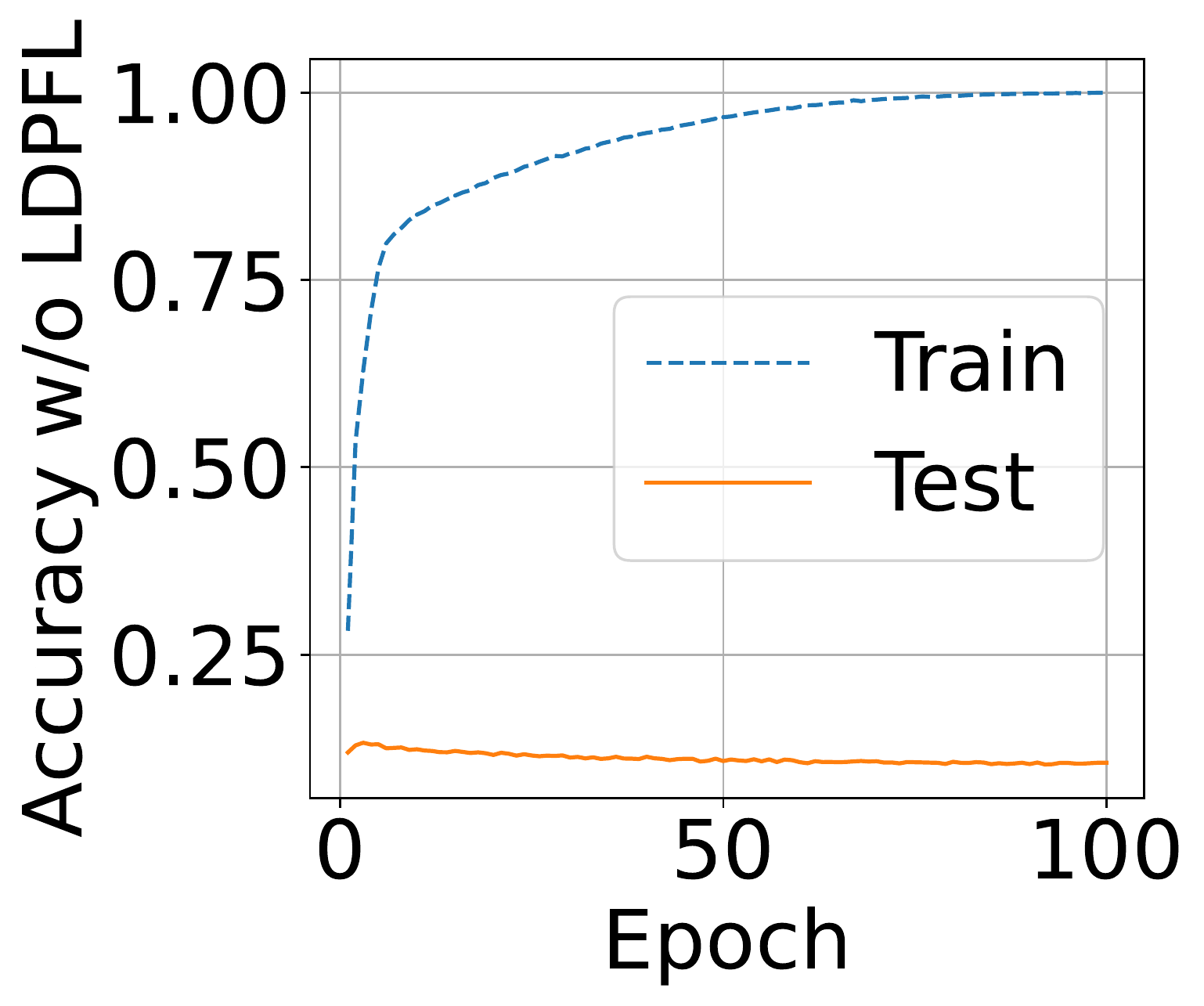}
			}
			\caption{CIFAR10 global model (refer to GM in Figures~\ref{ldpflarchitecture} and~\ref{ldpflcifar10}) performance. The third sub-figure shows the DNN accuracy without federated learning}\label{CIFAR10LDPFLperformance}
		\end{figure}
	\end{minipage}
	\vspace{-0.8cm}
	\begin{figure}[H] 
		\centering
		\subfloat[FMNIST local model performance of a client (randomly selected - refer to CNN in Figures~\ref{ldpflarchitecture} and~\ref{ldpflcifar10})]{\includegraphics[width=0.17\textwidth, trim=0cm 0cm 0cm 0cm]{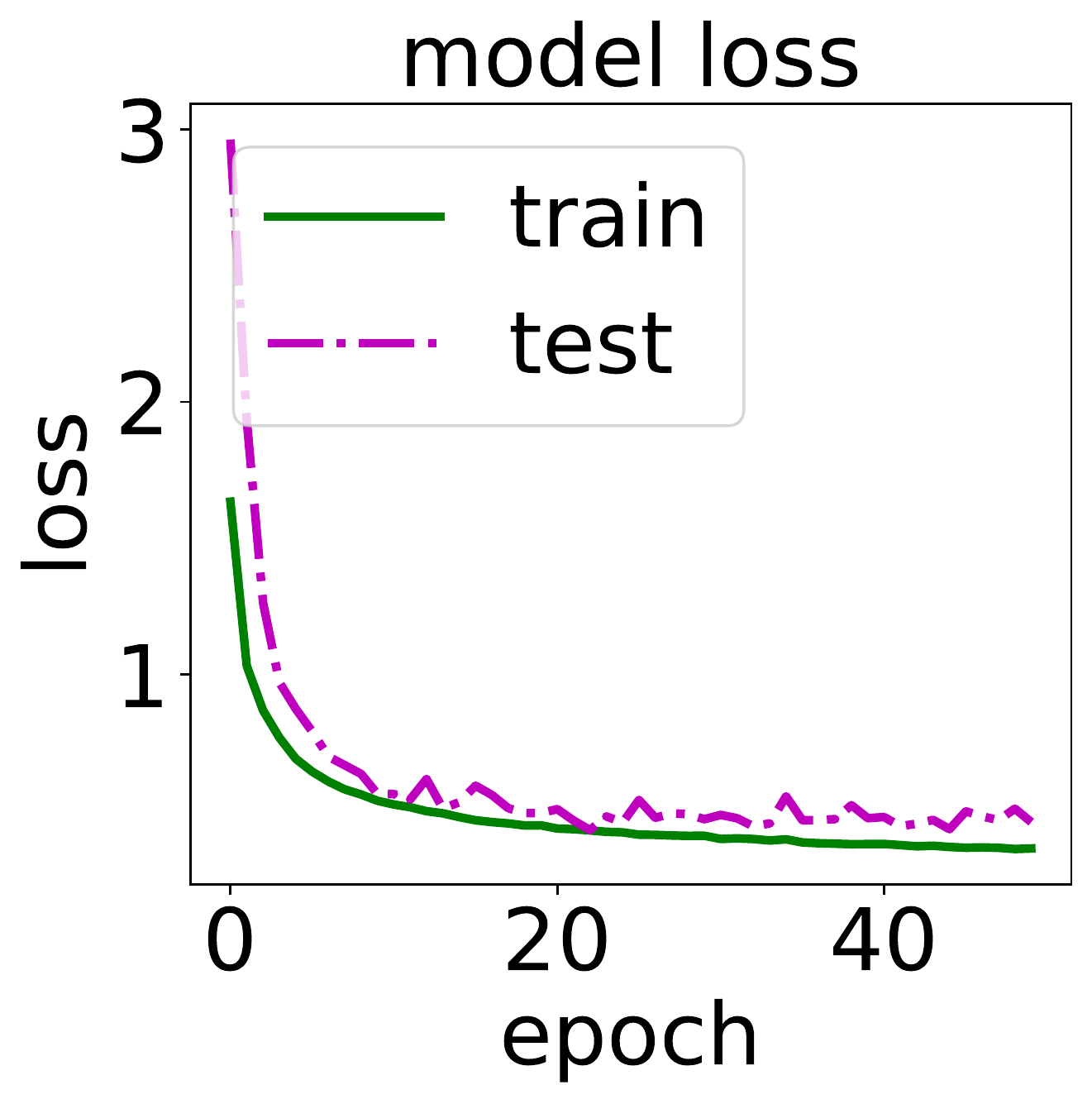}
			\includegraphics[width=0.19\textwidth, trim=0.3cm 0cm 0cm 0cm]{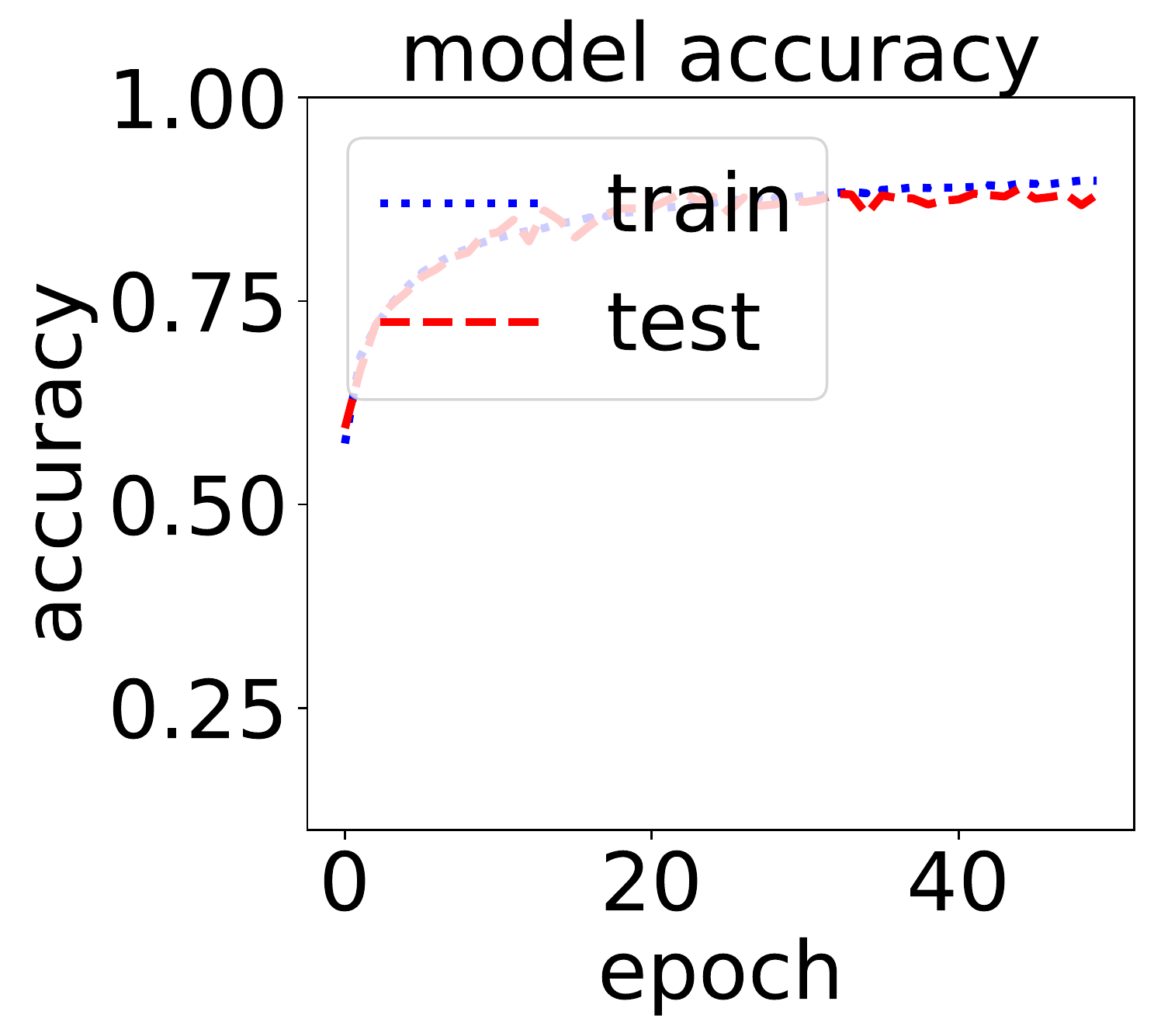}\
			\label{fmnistmodelperformance}
		}
		\hspace{1.5cm}
		\subfloat[FMNIST global model (refer to GM in Figures~\ref{ldpflarchitecture} and~\ref{ldpflcifar10}) performance]{\includegraphics[width=0.18\textwidth, trim=0cm 0cm 0cm 0cm]{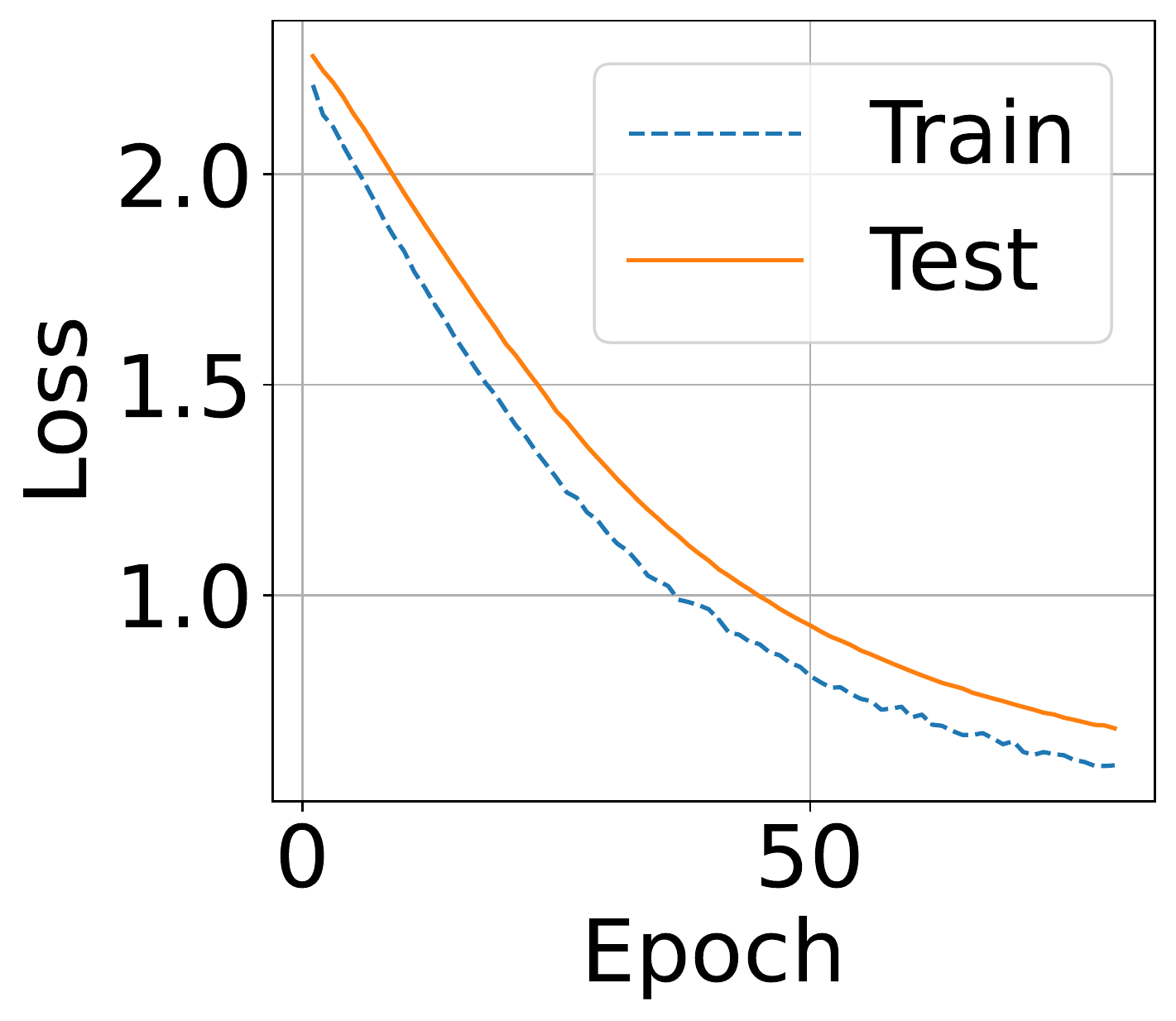}
			\includegraphics[width=0.18\textwidth, trim=0.3cm 0cm 0cm 0cm]{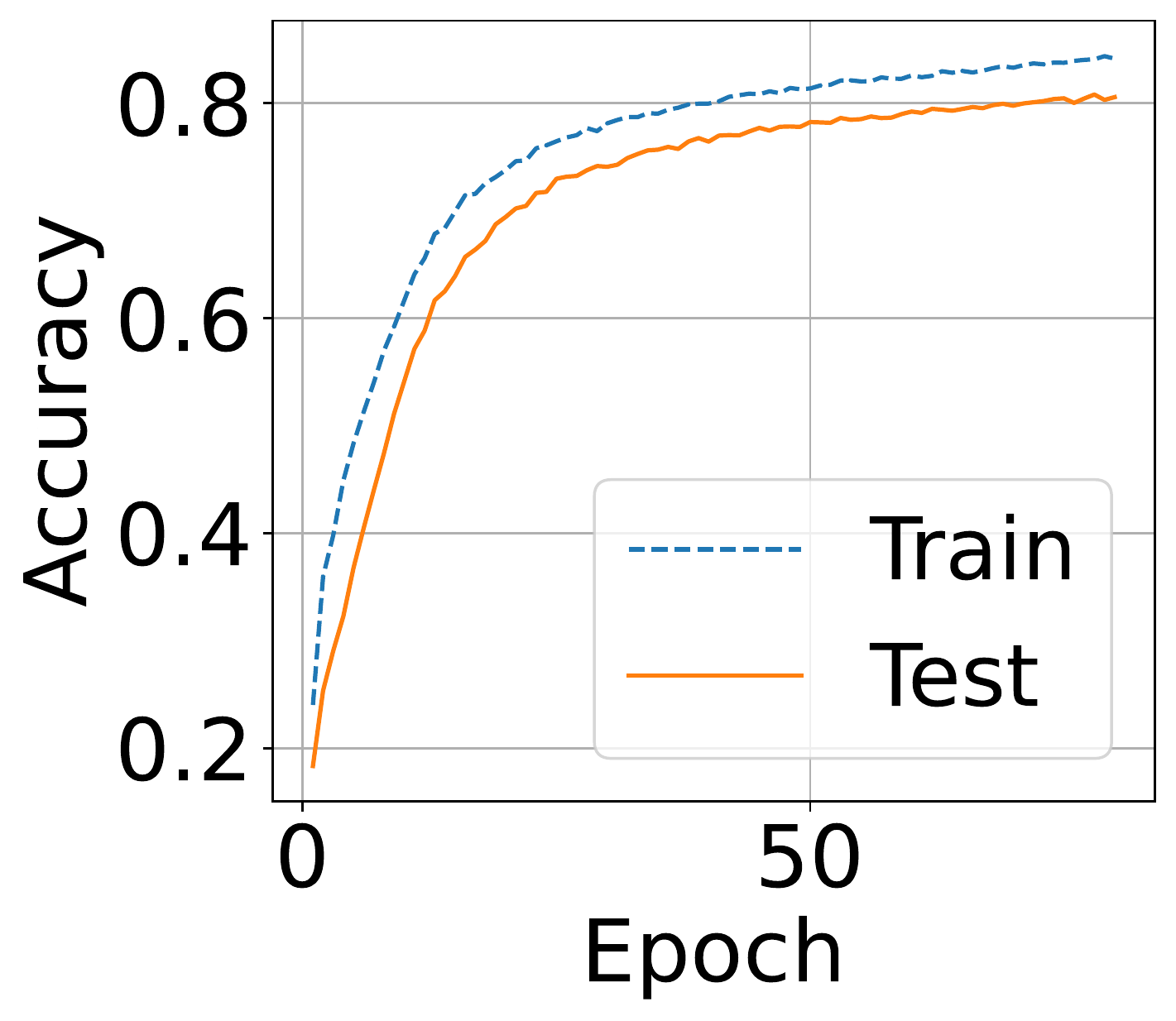}
			\label{FMNISTLDPFLperformance}
		}
		
		\caption{LDPFL performance under the FMNIST dataset}
	\end{figure}
	\vspace{-1.5cm}
	\begin{figure}[H] 
		\centering
		\subfloat[Data distribution among the clients (second and third plots represent the data distributions of two randomly selected clients)]{\includegraphics[width=0.60\textwidth, trim=0cm 0cm 0cm 0cm]{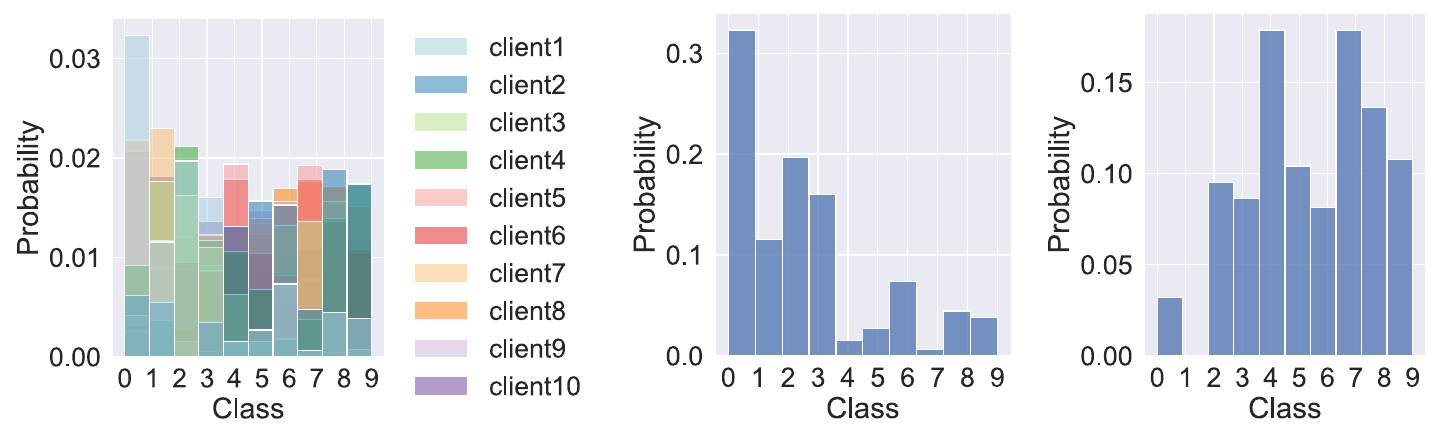}\label{FMNISTimbadist}}
		\hfill
		\subfloat[LDPFL Vs. vanilla FL performance]{\includegraphics[width=0.24\textwidth, trim=0.3cm 0cm 0cm 0cm]{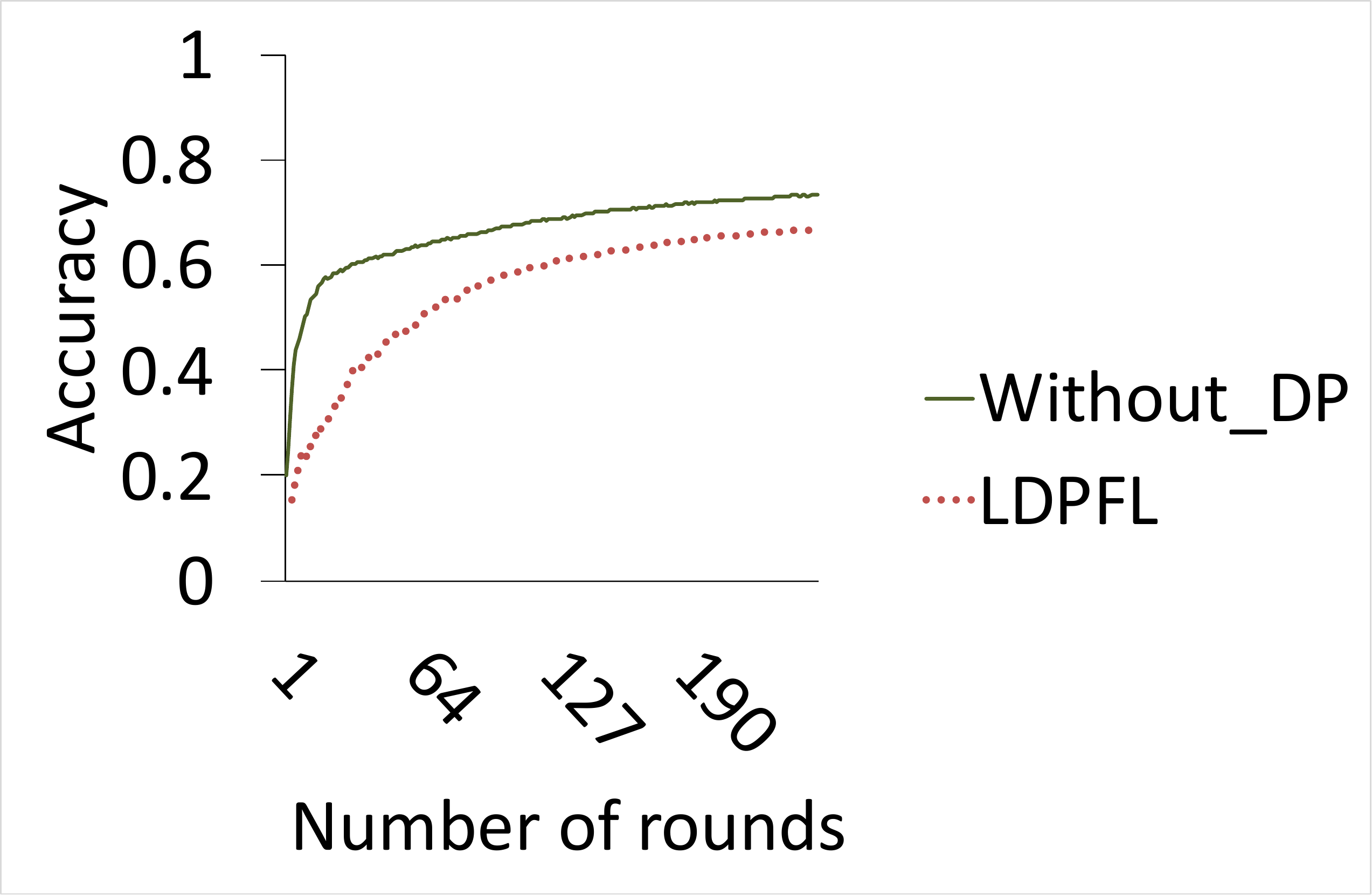}\label{twoclassimbdist}}
		\caption{Performance of LDPFL against vanilla FL under highly imbalanced (FMNIST) data (the non-IID setting)}
		\label{imbalancedata}
	\end{figure}
	\vspace{-1.2cm}
	\subsubsection{Randomizing data for differentially private FL}
	\label{difprivdata}
	
	The flattened outputs maintain a high correlation to the corresponding original inputs as the CM (refer to Fig.~\ref{ldplarchiflow}) was already trained on the inputs. Hence, the randomized data can appropriately preserve the input characteristics leading to good classification accuracy. During the data randomization, we maintained $m$, $n$, $\alpha$, and $\varepsilon$ at 4, 5, 10, and 0.5, respectively unless mentioned otherwise. With the sign bit, each digit in the flattened output is encoded to 10-bit $l = (m+n+1)$ binary representation. Since the sensitivity of an encoded binary string is equal to its length ($rl$), the binary strings generated under MNIST, CIFAR10, SVHN, and FMNIST  have sensitivities of $10240$, $20480$, $20480$, and $11520$, respectively. Hence, we maintain epsilon at 0.5, as increasing $\varepsilon$ within the acceptable limits (e.g., $0<\varepsilon\leq 10$) has a negligible impact on $\frac{\varepsilon}{{rl}/2}$. By maintaining $\alpha$ at a constant value of 10 (unless specified otherwise), we ensure that the binary string randomization dynamics are kept uniform during all experiments to observe unbiased results. However, to investigate the effect of $\alpha$ on the global model convergence, we changed $\alpha$ from 4 to 10 (refer to  Fig.~\ref{mnist_global_accuracy_comparison}). 
	
	
	\begin{minipage}[c]{0.29\linewidth}
		\begin{figure}[H] 
			\centering
			\includegraphics[width=0.46\textwidth, trim=0cm 0cm 0cm 0cm]{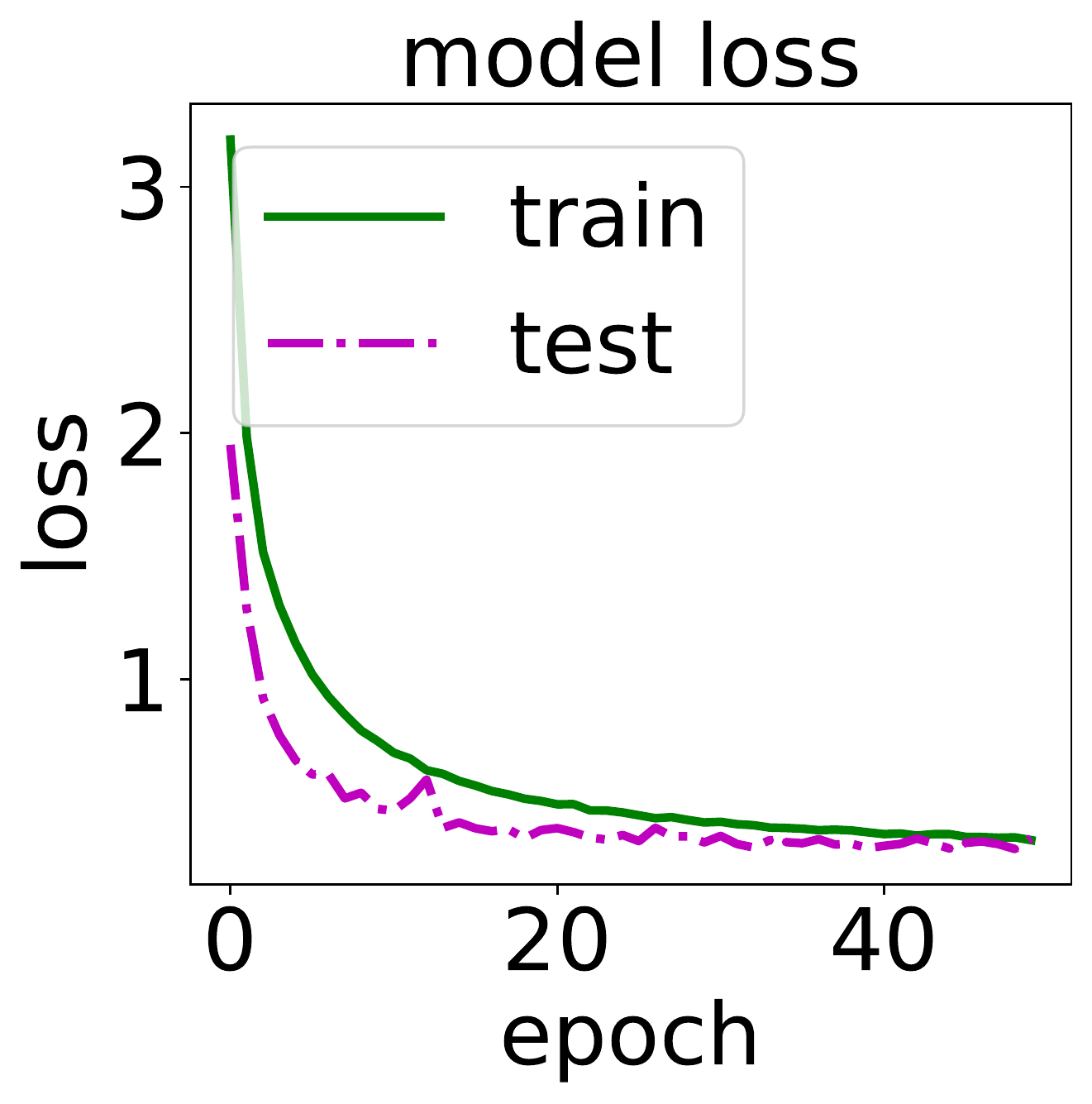}
			\hfill
			\includegraphics[width=0.51\textwidth, trim=0.3cm 0cm 0cm 0cm]{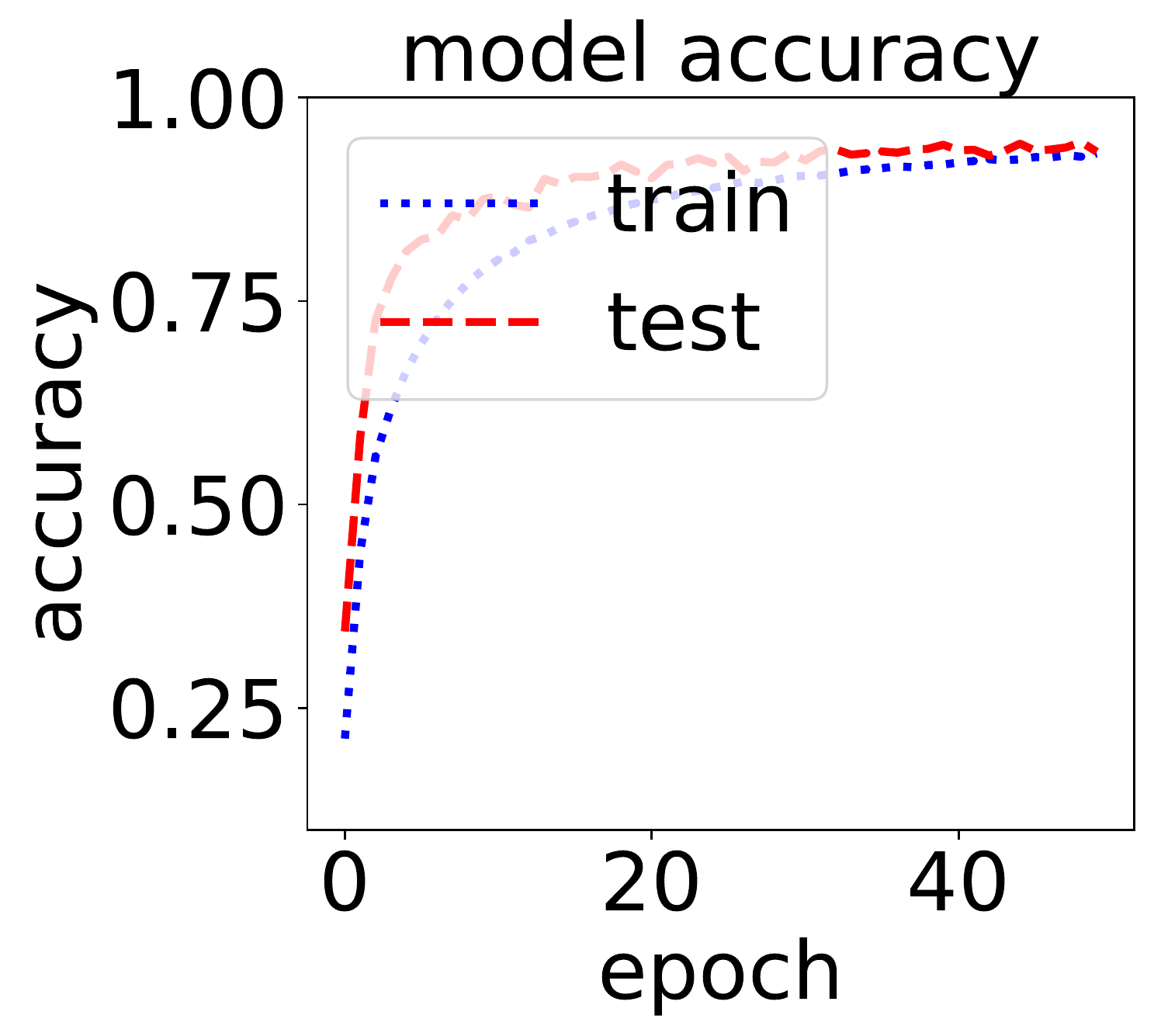}
			\caption{SVHN local model performance of a randomly selected client (refer to CNN in Figures~\ref{ldpflarchitecture} and~\ref{ldpflcifar10})}\label{svhnmodelperformance}
		\end{figure}
	\end{minipage}
	\hfill
	\begin{minipage}[c]{0.62\linewidth}
		\vspace{-0.7cm}
		\begin{figure}[H] 
			\centering
			\subfloat[2 clients]{\includegraphics[width=0.24\textwidth, trim=0cm 0cm 0cm 0cm]{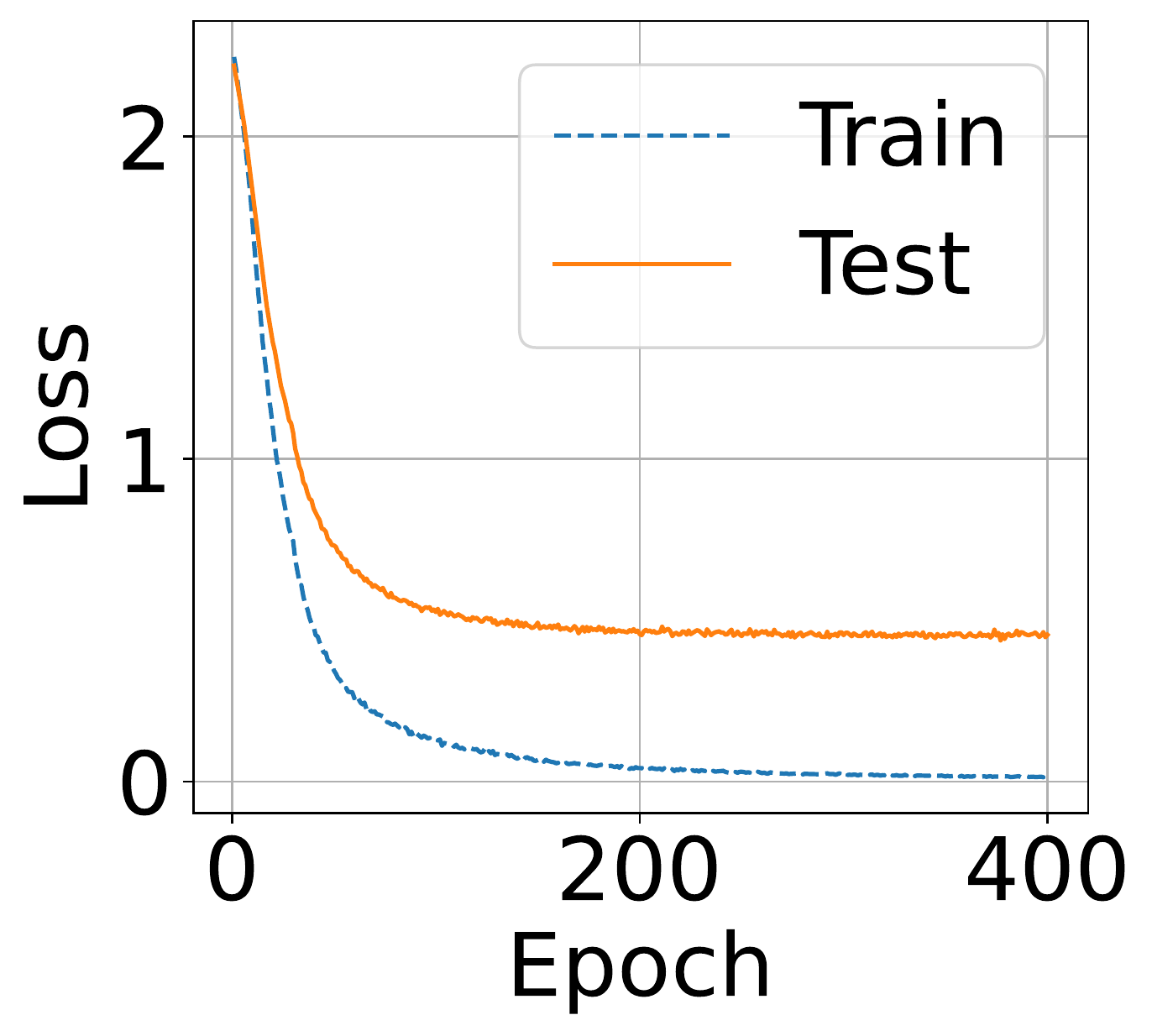}\label{svhnloss2clients}}
			\vspace{1mm}
			\subfloat[10 clients]{\includegraphics[width=0.24\textwidth, trim=0.3cm 0cm 0cm 0cm]{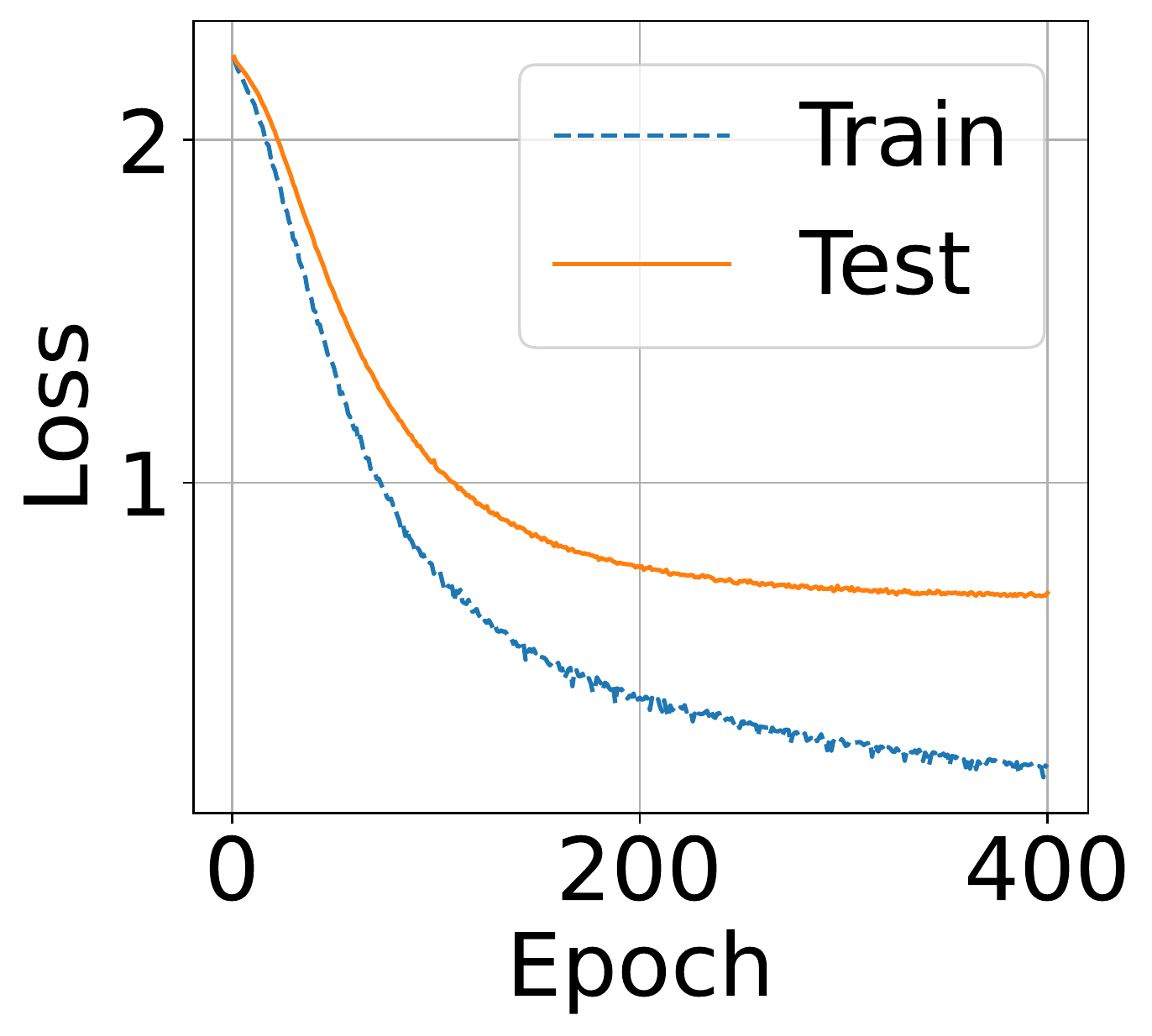}\label{svhnloss10clients}}
			\hfill
			\subfloat[20 clients]{\includegraphics[width=0.24\textwidth, trim=0cm 0cm 0cm 0cm]{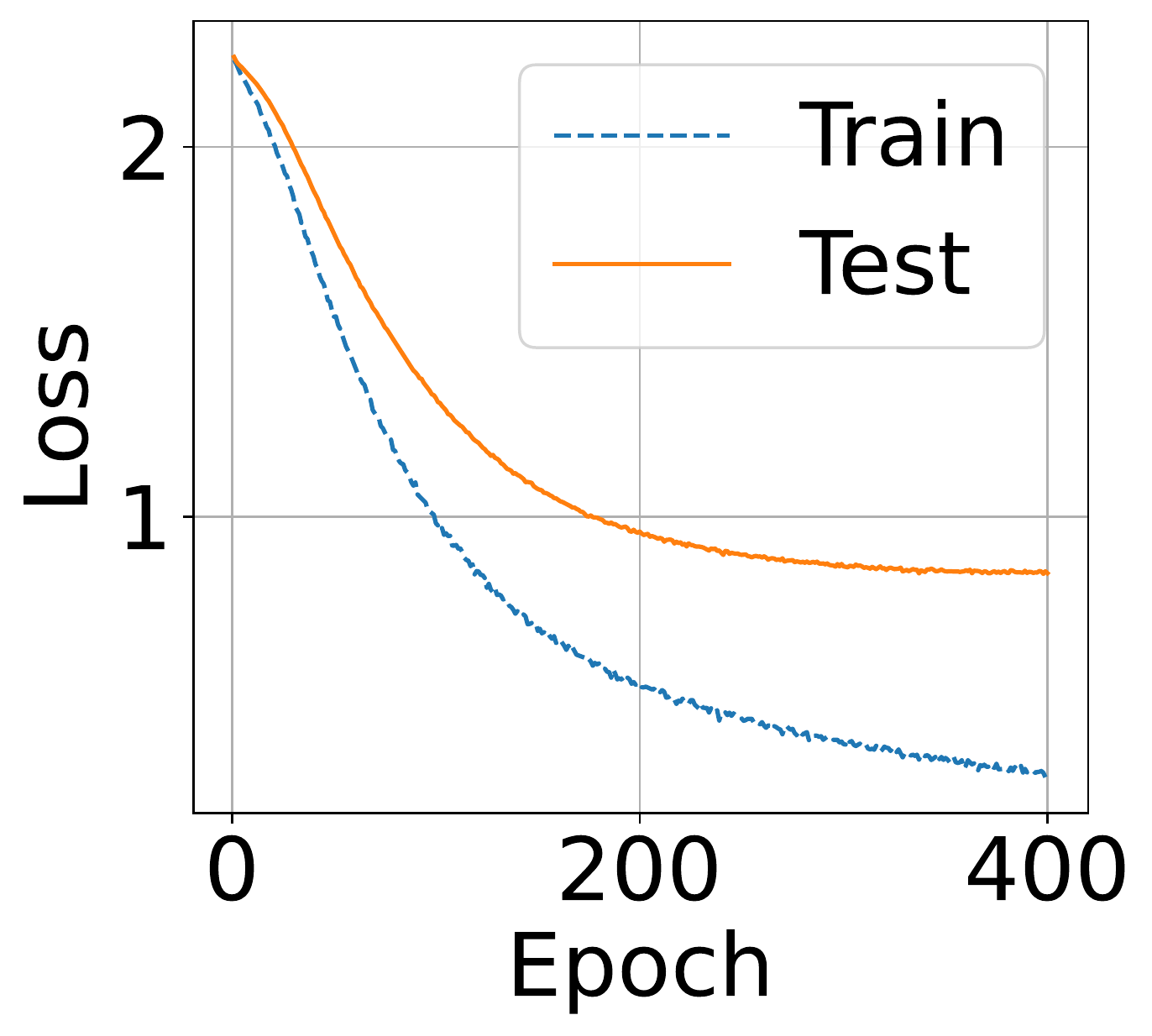}\label{svhnloss20clients}}
			\vspace{1mm}
			\subfloat[50 clients]{\includegraphics[width=0.25\textwidth, trim=0.3cm 0cm 0cm 0cm]{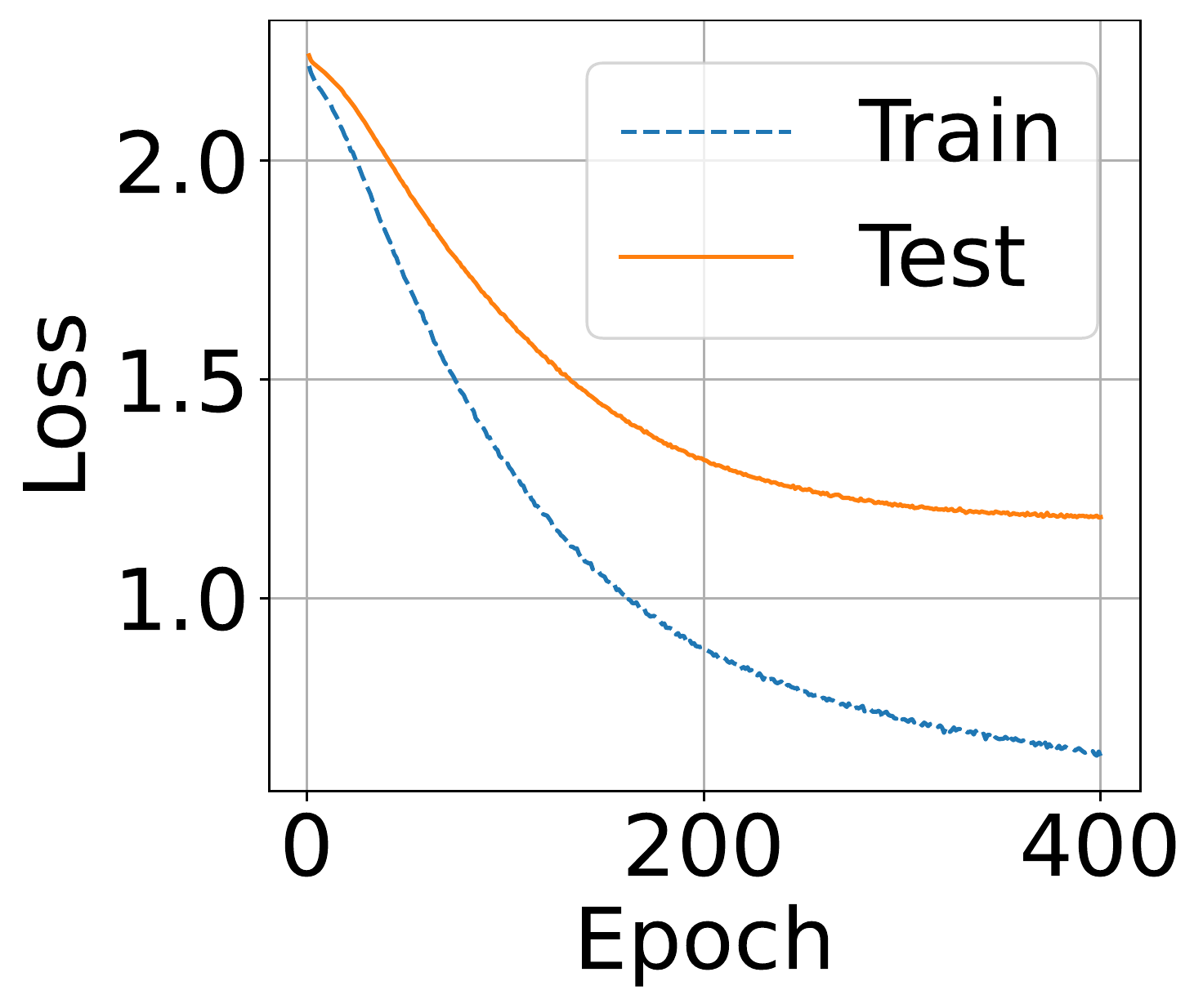}\label{svhnloss50clients}}
			\caption{SVHN global model (refer to GM in Fig.~\ref{ldpflarchitecture} and~\ref{ldpflcifar10}) loss under deferent numbers of clients}
			\label{svhnldpflmodelloss}
		\end{figure}
	\end{minipage}
	\vspace{-0.5cm}
	\begin{figure}[H] 
		\centering
		\scalebox{0.74}{
			\subfloat[2 clients]{\includegraphics[width=0.22\textwidth, trim=0cm 0cm 0cm 0cm]{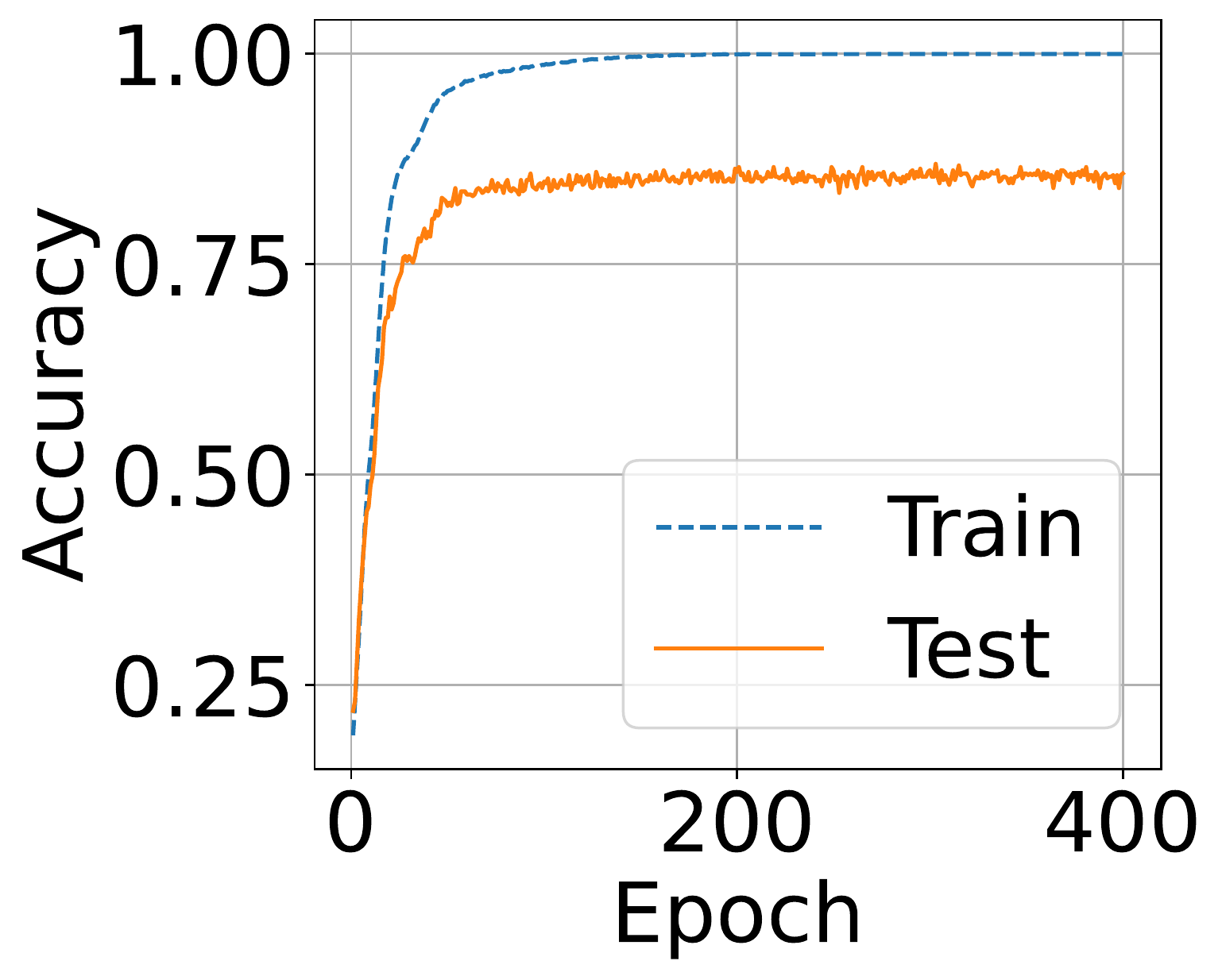}\label{svhnaccuracy2clients}}
			\vspace{1mm}
			\subfloat[10 clients]{\includegraphics[width=0.22\textwidth, trim=0.3cm 0cm 0cm 0cm]{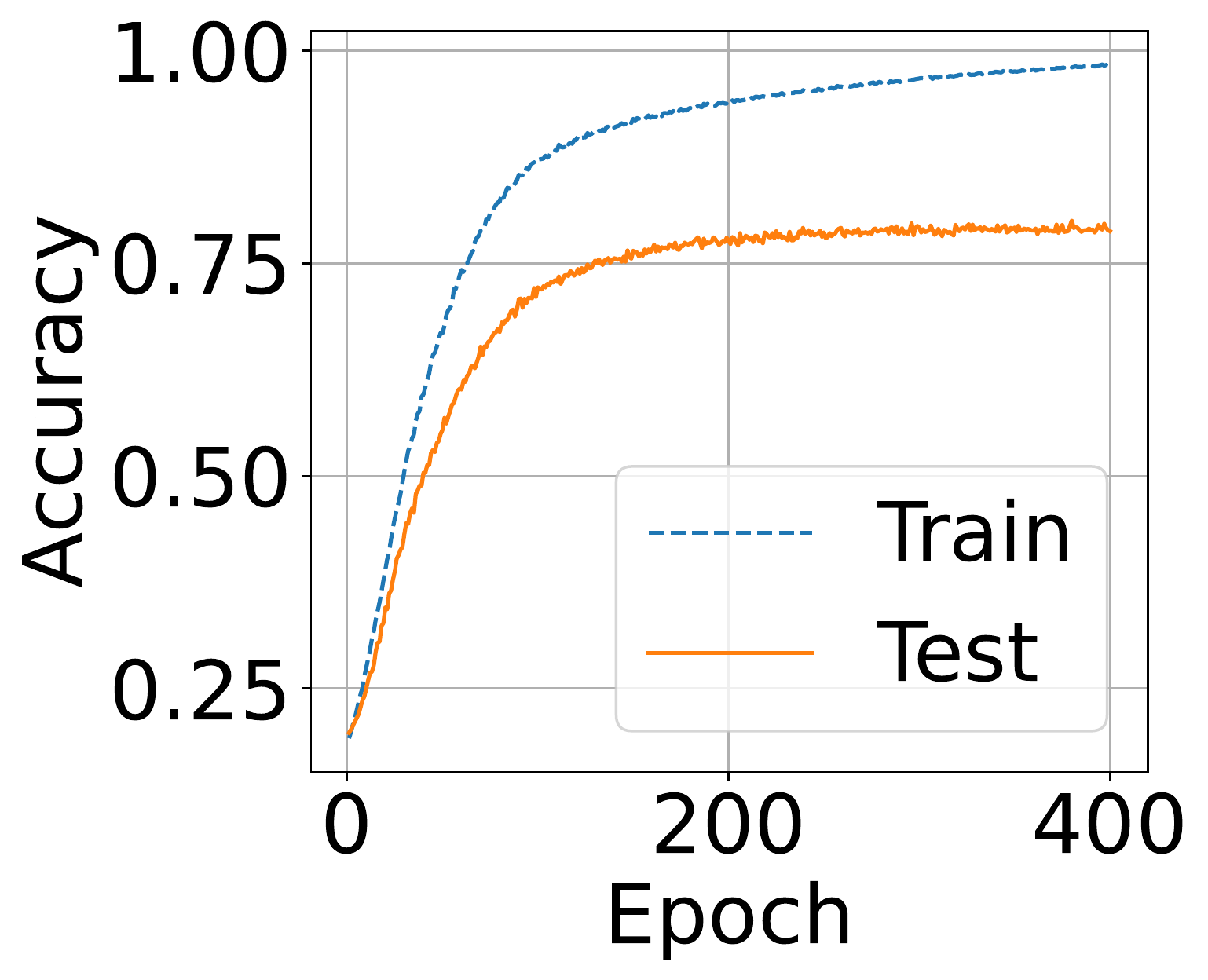}\label{svhnaccuracy10clients}}
			\vspace{1mm}
			\subfloat[20 clients]{\includegraphics[width=0.22\textwidth, trim=0cm 0cm 0cm 0cm]{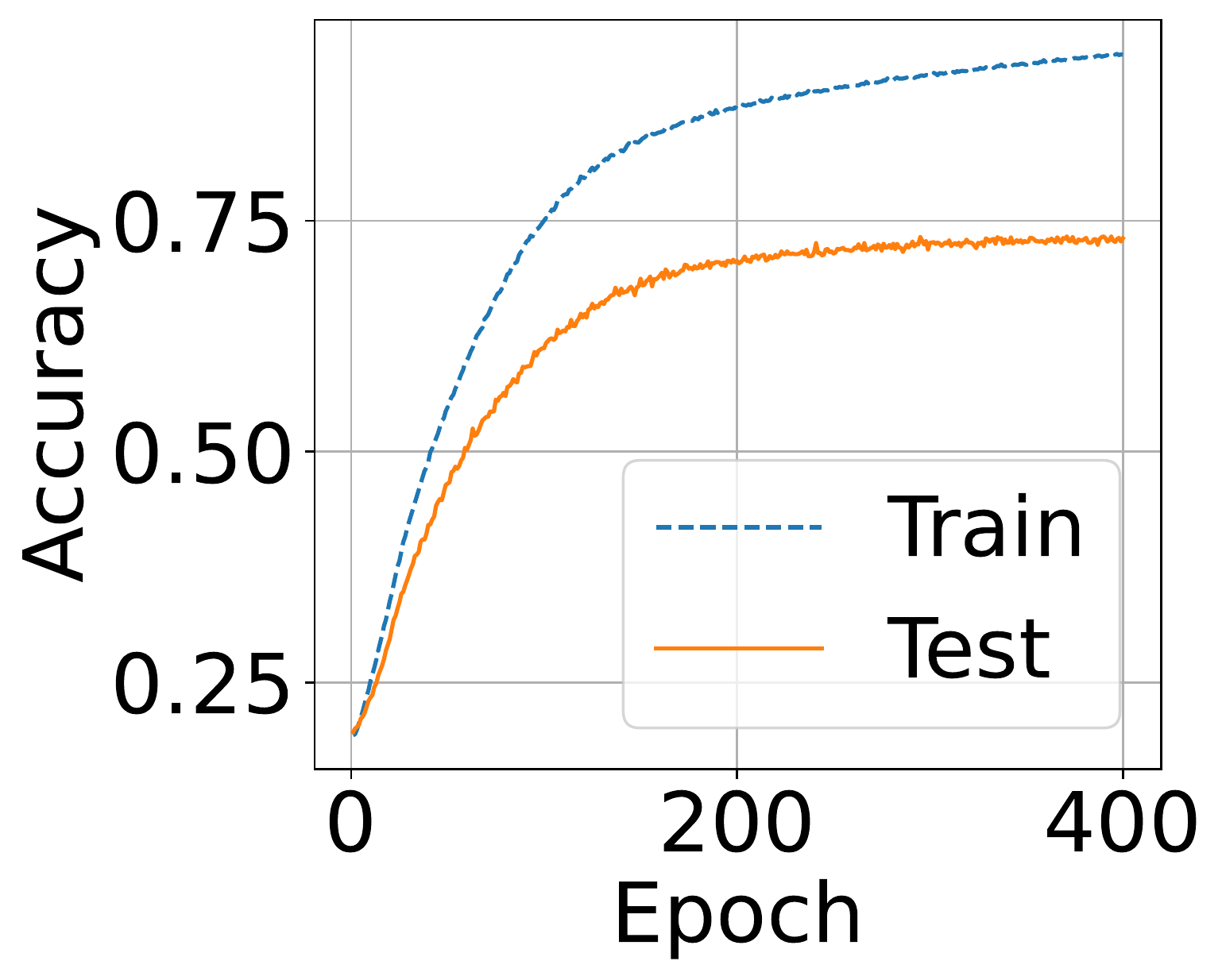}\label{svhnaccuracy20clients}}
			\vspace{1mm}
			\subfloat[50 clients]{\includegraphics[width=0.21\textwidth, trim=0.3cm 0cm 0cm 0cm]{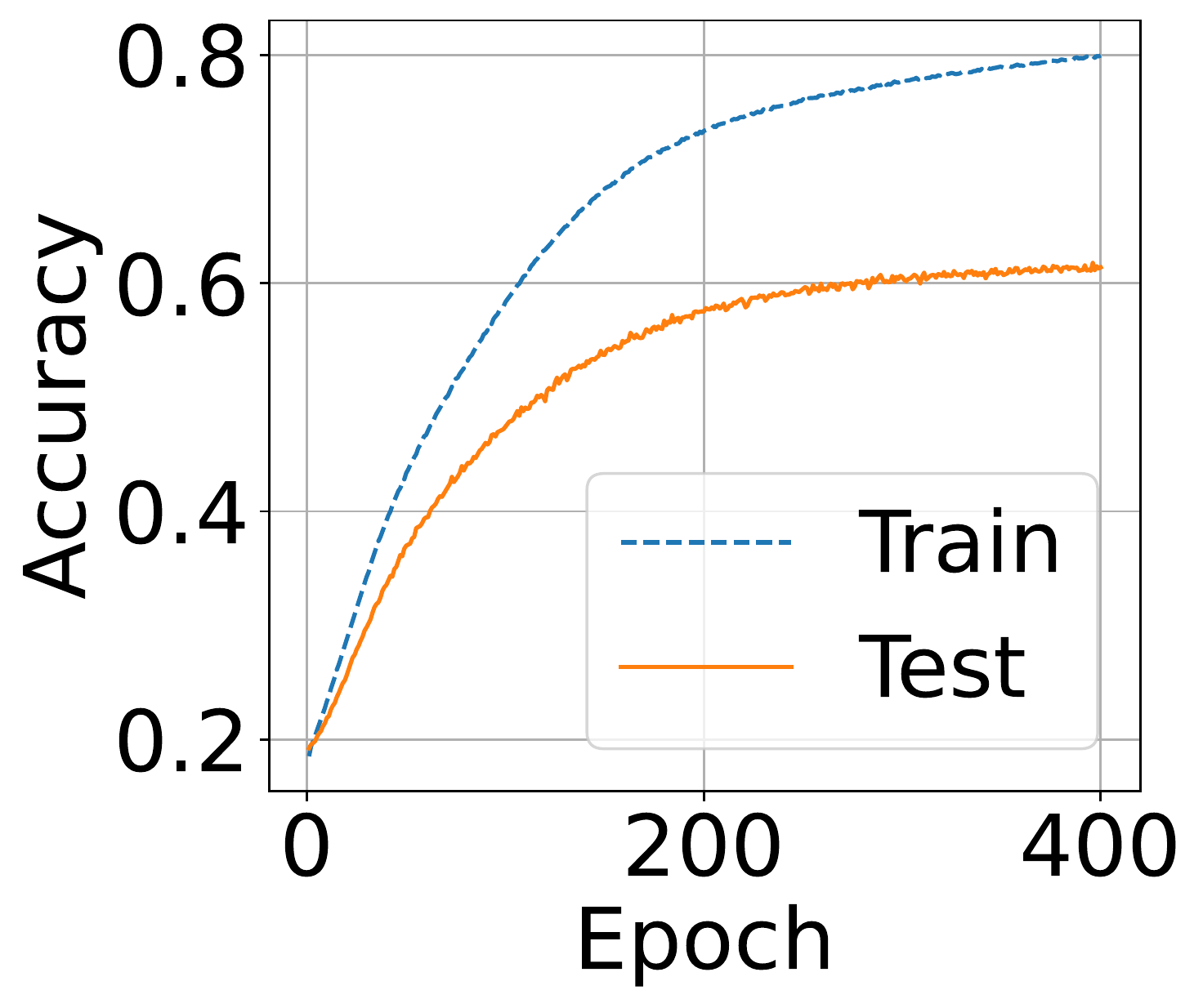}\label{svhnaccuracy50clients}}
		}
		\caption{SVHN global model (refer to GM in Figures~\ref{ldpflarchitecture} and~\ref{ldpflcifar10}) accuracy under deferent numbers of clients}
		\label{svhnldpflmodelaccuracy}
	\end{figure}
	\vspace{-1cm}
	
	\begin{multicols}{2}
		\begin{figure}[H]
			\centerline{\includegraphics[scale=0.11]{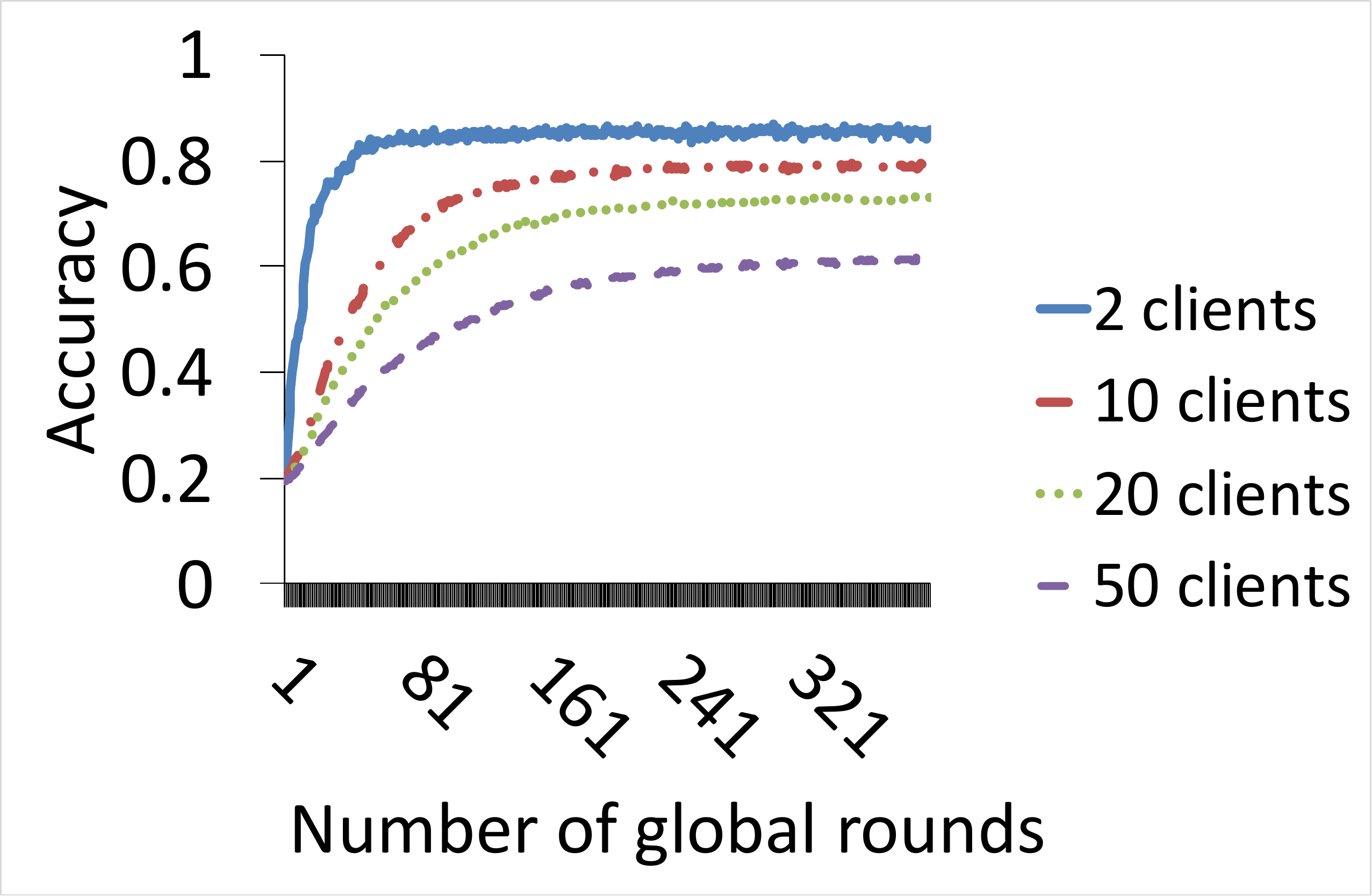}}
			\caption{SVHN Global model test accuracy comparison under different number of clients}
			\label{svhnmodelmulticlient}
		\end{figure}
		\hfil
		\begin{figure}[H]
			\centerline{\includegraphics[scale=0.11]{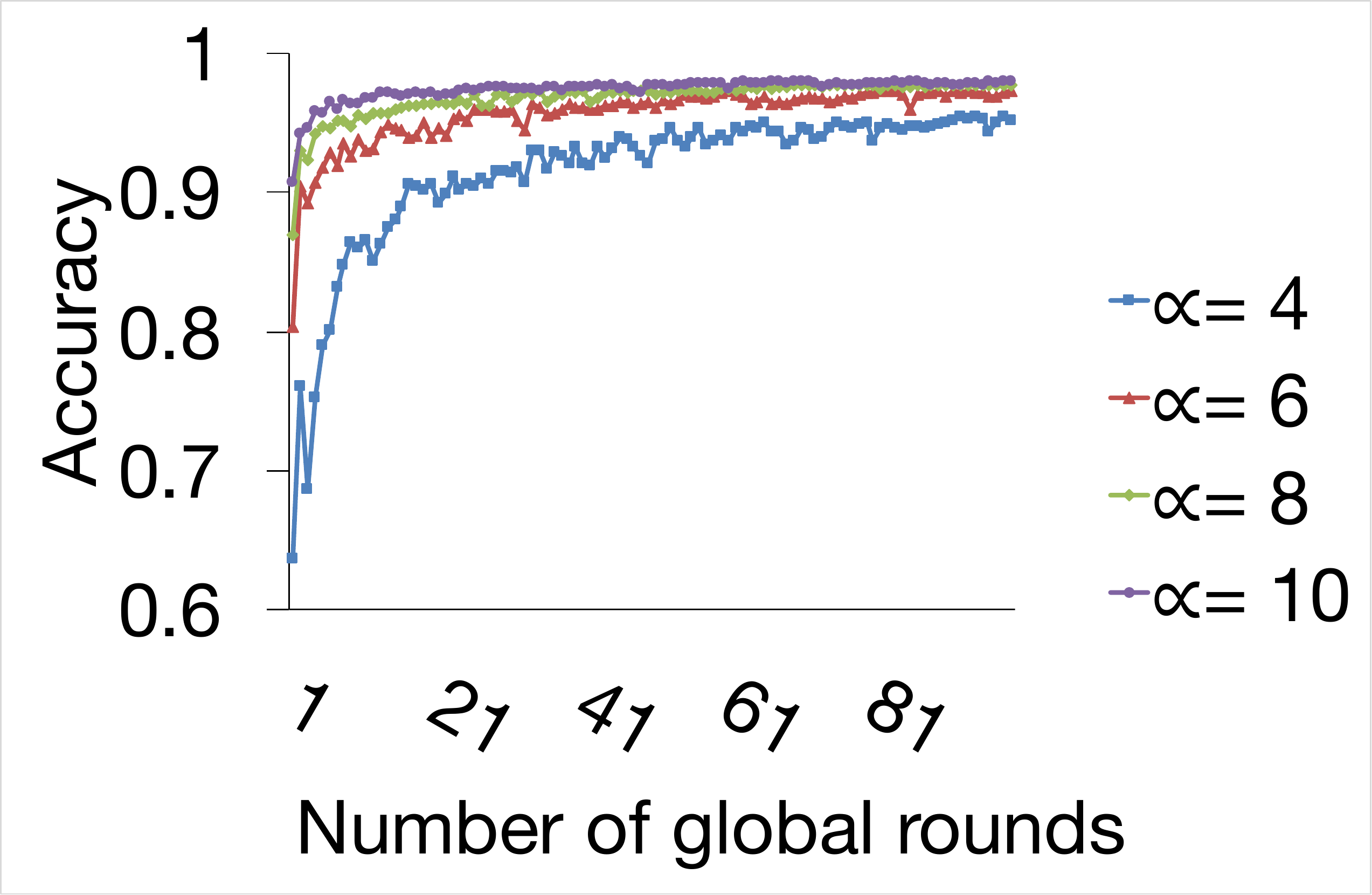}}
			\caption{MNIST Global model test accuracy comparison under different $\alpha$ values}
			\label{mnist_global_accuracy_comparison}
		\end{figure}
	\end{multicols}
	\vspace{-1.6cm}
	\subsubsection{Conducting federated learning over randomized data}
	\label{fedlearn}
	The part of the Figures~\ref{ldpflminst} and~\ref{ldpflcifar10} enclosed by the blue dotted square shows the configurations of the DNNs used in the FL setup of LDPFL. Under MNIST, all the clients use Adam (betas=(0.9, 0.999), eps=1e-08) optimizer, whereas, under CIFAR10, FMNIST, and SVHN, all the clients use stochastic gradient descent (SGD) optimizer for local model learning with a learning rate of 0.001 and a batch size of 32. Each client runs the local DNNs for 50 epochs and sends the trained parameters to the server. One round of FL includes executing client model training for 50 epochs, model federation, and model state update with federated parameters. We conduct different numbers of FL rounds sufficient to show the convergence patterns under each dataset based on the size. For FMNIST, we run FL for 80 rounds to replicate the settings of a previous study~\cite{truex2020ldp}, which we utilize for benchmarking.
	
	\vspace{-0.4cm}
	\subsection{LDPFL model performance}
	\vspace{-0.2cm}
	Figures~\ref{MNISTLDPFLperformance},~\ref{CIFAR10LDPFLperformance}, and~\ref{FMNISTLDPFLperformance} show the performance of the final LDPFL models under MNIST, CIFAR10, and FMNIST, respectively. LDPFL generates good model performance under both datasets. The global model performs well when the client models perform well, as evident from the plots. As LDPFL uses the fully trained CNN to generate a subsequent training dataset for the DP federated learning step of LDPFL, a good CNN client model enables producing a global model with good performance. As shown in the third sub-figure of Fig.~\ref{CIFAR10LDPFLperformance}, the client DNN is unable to generalize when the FL module is disabled, highlighting the importance of the LDPFL protocol. This shows that although the clients have good performing local CNN models, the client DNNs cannot generalize to learn features from other distributed entities without FL.   Figures~\ref{svhnldpflmodelloss} and~\ref{svhnldpflmodelaccuracy} show the model loss and accuracy convergence of LDPFL under different numbers of clients (under the SVHN dataset). Fig.~\ref{svhnmodelmulticlient} provides a comparison of the testing accuracy convergence of the LDPFL model under different client numbers (under the SVHN dataset). The higher the number of clients, the higher the time necessary for model convergence. We can also notice that the accuracy decreases when the number of clients increases, which reduces the total number of tuples within each client, producing CNNs with slightly less model performance. Consequently, each client applying LDP locally while maintaining local data representations can entail high randomization diversity. However, as shown in the plots, LDPFL provides a better approach to maintaining utility under complex datasets than other LDP approaches for lower privacy budgets when there are many clients. This is due to the clients in LDPFL maintaining the local data distributions by utilizing a locally converged model (with good performance) on the input data. Fig.~\ref{imbalancedata} shows the LDPFL performance under highly imbalanced data (the non-IID setting). According to the plots (refer to \ref{imbalancedata}), it is apparent that LDPFL follows (with reduced accuracy due to data randomization from DP) the convergence pattern of vanilla FL, confirming that the LDPFL algorithm does not impact the basic flow of the FL protocol. Fig.~\ref{mnist_global_accuracy_comparison} shows the performance of LDPFL under different levels of the privacy budget coefficient ($\alpha$). LDPFL takes slightly more time to converge with a slightly reduced accuracy when $\alpha$ is small. This is due to reduced $\alpha$ forcing LDPFL to increase the data randomization levels.

	\subsubsection{Performance comparison of LDPFL against existing approaches}
	For the performance comparison, we followed the benchmarking used in a previous study~\cite{truex2020ldp} on an approach named LDP-Fed that imposes $\alpha$-CLDP (a generalization of LDP~\cite{gursoy2019secure}) on federated learning. We compare the results of LDPFL against 4 previous approaches; (1) Non-private, (2) Secure multi-party computing (SMC)~\cite{bonawitz2017practical,truex2020ldp}, (3) Differentially private stochastic gradient descent (DPSGD)~\cite{abadi2016deep,truex2020ldp}, and (4) $\alpha$-Condensed Local Differential Privacy for Federated Learning ($\alpha$-CLDP-Fed)~\cite{truex2020ldp,gursoy2019secure}. These four approaches consider the k-Client selection protocol in which nine client updates will be considered for the federation in every round~\cite{truex2020ldp}. For benchmarking~\cite{truex2020ldp} set $\alpha$ of $\alpha$-CLDP-Fed to 1.0, and the privacy parameters (e.g., $\varepsilon$ and $\delta$) of the other three approaches are set accordingly to match with $\alpha=1.0$~\cite{truex2020ldp,gursoy2019secure}. We use the same default privacy parameters explained in Section~\ref{difprivdata} for LDPFL (refer to Section~\ref{difprivdata} for the primary factors that influence value assignments for the privacy parameters). The accuracy was generated on the FMNIST dataset. For LDPFL, we considered nine randomly chosen client updates out of 10. The model convergence of LDPFL for FMNIST is shown in Figures~\ref{FMNISTLDPFLperformance}. The accuracy values in Table~\ref{rescomp} are generated after 80 rounds of the federation. As shown in the table, LDPFL generates the second-highest accuracy. However, compared to LDPFL, $\alpha$-CLDP-Fed enforces a generalized form of LDP. Hence, LDPFL enforces the strictest privacy levels on the global model compared other four approaches (in Table~\ref{rescomp}), concluding that LDPFL delivers an overall better performance by providing a better balance between privacy and utility.
	\vspace{-0.8cm}
	\begin{table}[H]
		\caption{Comparison of LDPFL against the existing methods. \textbf{NA}: Not available, \textbf{ND}: Not defined, \textbf{Basic}: general DP (GDP), \textbf{Moderate}: not as strong as LDP but a generalization of LDP, which is better than general DP, \textbf{High}: satisfies strong LDP guarantees, \textbf{RQ}: Required, \textbf{NR}: Not required.}
		\centering
		\resizebox{0.9\columnwidth}{!}{
			\begin{tabular}{|l|l|l|l|l|l|}
				\hline
				\textbf{Method} & \textbf{\begin{tabular}[c]{@{}l@{}}Efficiency\\ (compared to \\ baseline)\end{tabular}} & \textbf{Privacy Model} & \textbf{\begin{tabular}[c]{@{}l@{}}Privacy Model\/ \\strength\end{tabular}} & \textbf{\begin{tabular}[c]{@{}l@{}}Trusted Party\\ Requirement\end{tabular}} & \textbf{\begin{tabular}[c]{@{}l@{}}Accuracy\\ (after 80 \\ rounds with\\ 9 client updates\\ every round)\end{tabular}} \\ \hline
				Non-private     & Baseline                                                                                & NA &            NA         & RQ                                                                           & $\sim$90\%                                                                                                             \\ \hline
				SMC             & Low                                                                                     & NA & ND                     & RQ                                                                           & $\sim$90\%                                                                                                             \\ \hline
				DPSGD           & High                                                                                    & ($\varepsilon$, $\delta$)-DP     & Basic         & RQ                                                                           & $\sim$80\%                                                                                                             \\ \hline
				$\alpha$-CLDP-Fed     & High                                                                                    & $\alpha$-CLDP   & Moderate             & NR                                                                           & $\sim$85.28\% - 86.93\%                                                                                                \\ \hline
				LDPFL           & High                                                                                    & $\varepsilon$-LDP   & High               & NR                                                                           & $\sim$81\%                                                                                                             \\ \hline
			\end{tabular}
		}
		\label{rescomp}
	\end{table}
	
	\vspace{-1cm}
	\section{Related Work}
	\vspace{-0.1cm}
	\label{relwork}
	Privacy-preserving approaches for FL can be broadly categorized into encryption-based (cryptographic)~\cite{bonawitz2017practical} and data modification-based (perturbation)~\cite{wei2020federated}. Cryptographic approaches look at how secure aggregation of parameters can be conducted at the FL server. The most widely adapted cryptographic approach for secure aggregation is secure multi-party computation (MPC)~\cite{fereidooni2021safelearn}. MPC enables the secure evaluation of a function on private data (also called secret shares) distributed among multiple parties who do not trust each other~\cite{bonawitz2017practical}. The requirement of a trusted party (e.g., VerifyNet~\cite{xu2019verifynet}, and VeriFL~\cite{guo2020v}) or the requirement of a considerably high number of communications (e.g., Bonawitz et al.'s approach~\cite{bonawitz2017practical} and Bell et al.'s approach~\cite{bell2020secure}) are two of the fundamental problems of most of the existing MPC approaches for FL~\cite{bonawitz2017practical}. Besides, the existing MPC approaches show vulnerability towards advanced adversarial attacks such as backdoor attacks~\cite{bagdasaryan2020backdoor}. Homomorphic encryption (HE) is the other frequently adapted cryptographic approach for the secure aggregation of parameters in FL. HE enables algebraic operations over encrypted data to produce a ciphertext that can be decrypted to obtain the algebraic outcome on the original plaintext with security and privacy~\cite{gentry2009fully}. However, scalability has been a major challenge in HE. The latest approaches, such as BatchCrypt, try to introduce less complex HE-based solutions for secure FL parameter aggregation~\cite{zhang2020batchcrypt}. Besides, the distributed setting makes HE infeasible for large-scale scenarios due to the low efficiency~\cite{yang2019federated,so2021turbo}. Both global differential private (GDP) approaches ~\cite{geyer2017differentially,asoodeh2021differentially} and local differential private (LDP)~\cite{truex2020ldp,seif2020wireless} approaches were introduced to FL~\cite{wei2020federated}. GDP approaches focus on privately learning the algorithm (e.g., SGD)~\cite{geyer2017differentially,mcmahan2017learning}, whereas LDP approaches~\cite{truex2020ldp,seif2020wireless} focus on randomizing the data inputs to the algorithm (it can be the direct randomization of user inputs or randomization of the model parameters before sending them to the aggregator) to learn on randomized data. Robin et al.'s approach~\cite{geyer2017differentially} and Asoodeh et al.'s approach~\cite{asoodeh2021differentially} are two of the GDP approaches for FL, whereas LDP-Fed~\cite{truex2020ldp} and Seif et al.'s approach~\cite{seif2020wireless} are two LDP approaches. The primary issue of most GDP approaches is the requirement of a trusted aggregator. These approaches focus more on privacy leaks among the FL clients~\cite{geyer2017differentially,asoodeh2021differentially}. By either randomizing user inputs or parameters before sending them to the aggregator, LDP-based approaches provide a stricter privacy setting~\cite{seif2020wireless,truex2020ldp}. However, existing LDP approaches often consume unreliable privacy budgets ($\varepsilon$) to produce good accuracy, work on generalized LDP guarantees (e.g., $\alpha$-CLDP ), or do not produce a high accuracy compared to GDP approaches. Hence, there is a significant imbalance between the privacy and utility of LDP approaches. Developing new LDP approaches, such as LDPFL, is essential to answer these challenges.
	\vspace{-0.2cm}
	\section{Conclusion}
	\vspace{-0.2cm}
	\label{conclusion}
	We proposed a utility-enhancing, differentially private federated learning approach (abbreviated as LDPFL) for industrial (cross-silo) settings. LDPFL uses local differential privacy (LDP) to enforce strict privacy guarantees on FL. The proposed approach provides high testing accuracy (e.g., 98\%) under strict privacy settings (e.g., $\varepsilon=0.5$). LDPFL preserves data utility by using a fully trained local model to filter and flatten the input features. The LDPFL's LDP model enables high utility preservation by randomizing one half of the binary string differently from the other half, ensuring a high bit preservation during binary string randomization. The LDP approach of LDPFL also allows federated learning under untrusted settings (e.g., with untrusted clients and an untrusted server) while preserving high privacy and utility. Besides, benchmarking suggests that LDPFL is preferred when high utility is required under strict privacy settings (maintaining a proper balance between privacy and utility). 
	
	\newpage
	
	\section*{Acknowledgment}
	The work has been supported by the Cyber Security Research Centre Limited whose activities are partially funded by the Australian Government’s Cooperative Research Centres Programme.
	
	\section*{Appendices}
	
	\subsection*{Appendix A: Model configurations}
	\label{cnnconfig}
	\vspace{-0.8cm}
	\begin{figure}[H] 
		\centering
		\scalebox{1.1}{
			\subfloat[LDPFL architecture used for the CIFAR10 and SVHN datasets.  Act = Activation, BatchNorm = Batch normalization, RND layer= Randomization layer, FL = Federated learning. \textbf{Note}: When the input dataset is FMNIST, the CNN layer 1 size was changed to 28x28x1, and the DNN and GM layer 1 sizes were changed to 11520.]{\includegraphics[width=0.5\textwidth, trim=0.3cm 0cm 0cm 0cm]{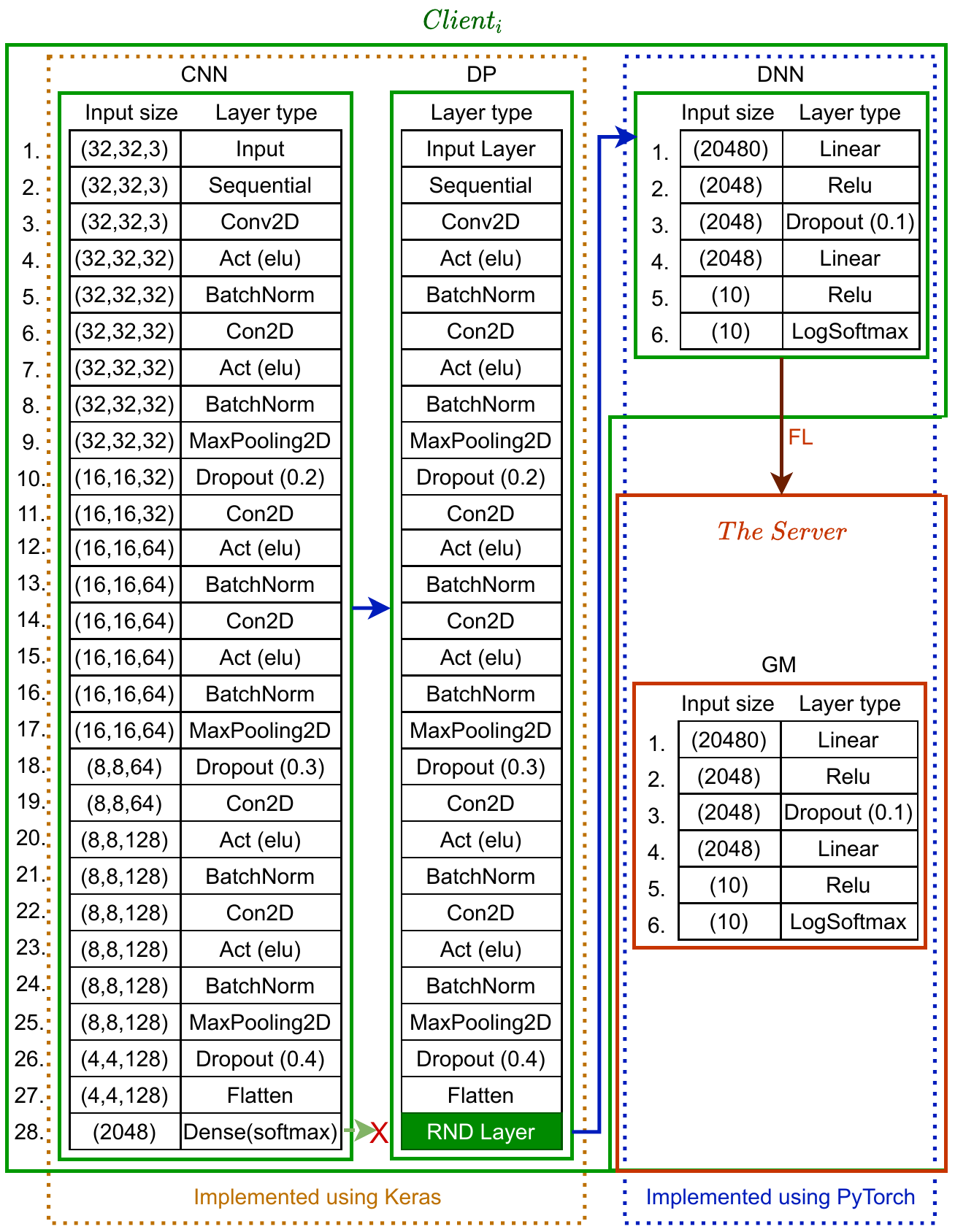}\label{ldpflcifar10}}
			\vspace{4mm}
			\subfloat[LDPFL architecture used for the MNIST dataset. Act = Activation, BatchNorm = Batch normalization, RND layer= Randomization layer, FL = Federated learning.]{\includegraphics[width=0.4\textwidth, trim=0cm 0cm 0cm 0cm]{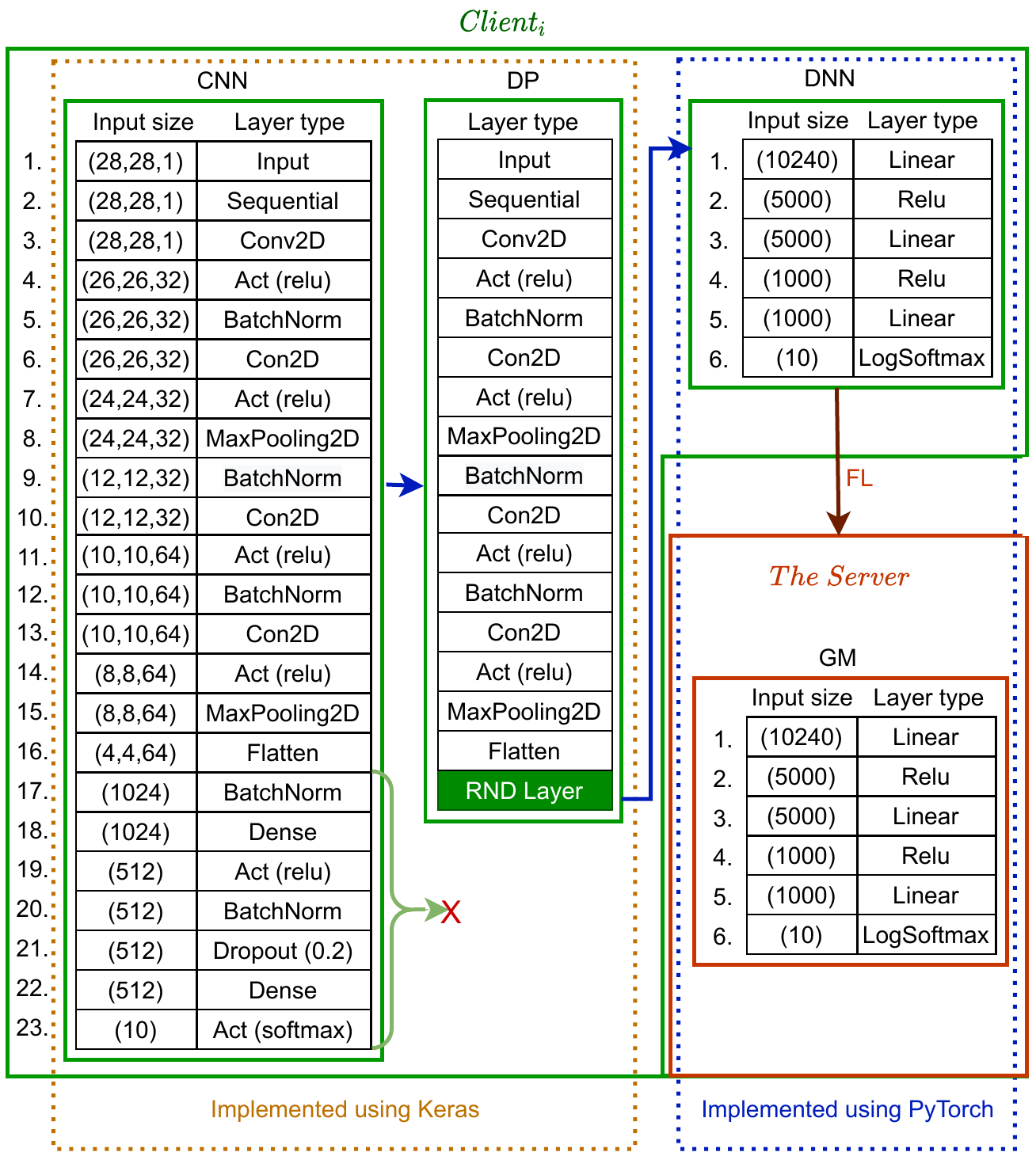}\label{ldpflminst}}
		}
		\caption{LDPFL architectures used for the datasets}
		\label{archiconfig}
	\end{figure}
	
	The images in the MNIST dataset have a resolution of 28x28x1 (one channel), which are size-normalized and centered~\cite{lecun1998gradient}. Hence, the input layer size of the CNN used for MNIST is 28x28x1 (refer to Fig.~\ref{ldpflminst}. Convolution layers no. 3 and no. 6 use 32, 3 $\times$ 3 filters with stride 1, whereas convolution layers no.10 and no.13 use  64, 3 $\times$ 3 filters with stride 1. We used a kernel regularizer of regularizers.l2(weight\_decay = 1e-4) for all convolution layers. Both max-pooling layers (layers no. 8 and no.15) use 2$\times$2 max pools. All batch normalization layers (layer numbers 5,9,12,17, and 20) use ``axis=-1''.
	
	The images in the CIFAR10 and SVHN datasets have a resolution of 32x32x3, which are size-normalized and centered~\cite{lecun1998gradient}. Hence, the input layer size of the CNNs used for CIFAR10 and SVHN is 32x32x3 (refer to Fig.~\ref{ldpflcifar10}). Convolution layers no. 3 and no. 6 use 32, 3 $\times$ 3 filters with stride 1, convolution layers no.11 and no.14 use  64, 3 $\times$ 3 filters with stride 1, and convolution layers no.19 and no.22 use  128, 3 $\times$ 3  filters with stride 1. All three max-pooling layers (layers no. 9, no.17, and no.25) use 2$\times$2 max pools. The image resolution of FMNIST images is 28x28x1. Hence, only the input layer size of the CNN (refer to Fig.~\ref{ldpflcifar10}) was changed to 28x28x1 while keeping all other settings of the local CNN architecture unchanged.

\end{document}